\documentclass[11pt,letterpaper]{article}

\usepackage[T1]{fontenc}
\usepackage{lmodern}
\usepackage[utf8]{inputenc}
\usepackage{amsmath,amssymb,amsthm,mathtools,bm}
\usepackage{booktabs,array}
\usepackage{graphicx}
\usepackage{microtype}
\usepackage{enumitem}
\usepackage{float}
\usepackage[margin=1in]{geometry}
\usepackage[hidelinks]{hyperref}
\usepackage[nameinlink,noabbrev]{cleveref}

\title{Geometric and Arithmetic Likelihood Aggregation for Diffusions with Heterogeneous Volatility}
\author{Jan Vecer\thanks{Department of Probability and Mathematical Statistics, Charles University, Prague, Czech Republic\\ (\href{mailto:vecer@karlin.mff.cuni.cz}{\texttt{vecer@karlin.mff.cuni.cz}}).}}

\date{September 8, 2026}
\hypersetup{
  pdftitle={Geometric and Arithmetic Likelihood Aggregation for Diffusions with Heterogeneous Volatility},
  pdfauthor={Jan Vecer}
}

\theoremstyle{plain}
\newtheorem{theorem}{Theorem}[section]
\newtheorem{proposition}[theorem]{Proposition}
\newtheorem{lemma}[theorem]{Lemma}
\newtheorem{corollary}[theorem]{Corollary}
\newtheorem{assumption}[theorem]{Assumption}
\theoremstyle{definition}
\newtheorem{definition}[theorem]{Definition}
\theoremstyle{remark}
\newtheorem{remark}[theorem]{Remark}

\crefname{theorem}{Theorem}{Theorems}
\Crefname{theorem}{Theorem}{Theorems}
\crefname{proposition}{Proposition}{Propositions}
\Crefname{proposition}{Proposition}{Propositions}
\crefname{lemma}{Lemma}{Lemmas}
\Crefname{lemma}{Lemma}{Lemmas}
\crefname{corollary}{Corollary}{Corollaries}
\Crefname{corollary}{Corollary}{Corollaries}
\crefname{definition}{Definition}{Definitions}
\Crefname{definition}{Definition}{Definitions}
\crefname{assumption}{Assumption}{Assumptions}
\Crefname{assumption}{Assumption}{Assumptions}
\crefname{remark}{Remark}{Remarks}
\Crefname{remark}{Remark}{Remarks}
\crefname{example}{Example}{Examples}
\Crefname{example}{Example}{Examples}

\newcommand{\R}{\mathbb{R}}
\newcommand{\E}{\mathbb{E}}
\newcommand{\Pp}{\mathbb{P}}

\newcommand{\F}{\mathcal{F}}
\newcommand{\A}{\mathcal{A}}

\newcommand{\D}{\mathfrak{D}}
\newcommand{\ellit}{\mathfrak{l}}
\newcommand{\KL}{D_{\mathrm{KL}}}
\newcommand{\AW}{\mathcal{AW}}
\newcommand{\SW}{\mathcal{SW}}
\newcommand{\BW}{d_{\mathrm{BW}}}
\newcommand{\tr}{\operatorname{tr}}

\newcommand{\argmin}{\operatorname*{arg\,min}}
\newcommand{\dd}{\,\mathrm{d}}

\newcommand{\law}{\mathcal{L}}
\newcommand{\Spp}{\mathbb{S}_{++}^{d}}
\newcommand{\Sp}{\mathbb{S}_{+}^{d}}
\newcommand{\norm}[1]{\left\lVert #1\right\rVert}

\newcommand{\Fid}{\operatorname{F}}

\numberwithin{equation}{section}
\allowdisplaybreaks

\begin{document}
\maketitle

\begin{abstract}
We study how to combine diffusion models that disagree about drift and covariance. Candidate-first relative-entropy minimization gives geometric pooling, whereas expert-first minimization gives the arithmetic mixture associated with weighted logarithmic wealth. Different quadratic variations can make path-space entropy infinite, and the arithmetic mixture need not be a Markov diffusion. We therefore specify a local criterion combining drift information, normalized by the second argument's covariance, with quadratic transport between Gaussian shocks in a fixed state metric. A Gaussian identity and an Euler convergence estimate justify this chosen criterion. The expert-first projection has posterior-mean drift and an inverse-covariance penalty for drift dispersion; in one dimension this penalty increases volatility. For Ornstein--Uhlenbeck experts with a common mean-reversion rate, coefficient regularity holds on the full horizon for common volatility and away from the initial time for heterogeneous volatilities. The candidate-first problem has a Hamilton--Jacobi--Bellman characterization. Its matrix covariance selector reduces by congruence to a Bures--Wasserstein barycenter. The condition $H+\lambda M\succ0$, with value Hessian $H$ and state metric $M$, is sharp for finiteness of the unrestricted local covariance problem; compact constraints keep that problem finite. A covariance-disagreement budget interprets the penalty parameter. Linear--quadratic, exact-transition, and financial examples distinguish dynamic volatility reduction, drift-dispersion inflation, and martingale restrictions.
\end{abstract}

\noindent\textbf{Keywords:} Kullback--Leibler divergence, diffusion aggregation, Bures--Wasserstein distance, stochastic volatility, stochastic control

\noindent\textbf{2020 Mathematics Subject Classification:} 49L25, 49Q22, 60H10, 91G20, 94A17

\section{Introduction}\label{sec:intro}

Suppose several diffusion models describe the same state variables but disagree about drift and covariance. How should they be combined into one model for a decision problem? In finance, the state may contain asset prices and volatility factors; in economics, it may describe output, consumption, or aggregate risk. Averaging coefficients gives a diffusion, but does not identify the approximation problem it solves. We formulate two such problems and characterize their optimal coefficients.

The distinction begins with the order of relative entropy. For expert laws $P_k$ and positive weights $\pi_k$ summing to one, the classical identities are
\begin{align}
 \argmin_Q\sum_k\pi_k\KL(Q\|P_k)
 &\quad\text{is the normalized geometric pool},
 \label{eq:intro_forward_pool}\\
 \argmin_Q\sum_k\pi_k\KL(P_k\|Q)
 &\quad\text{is the arithmetic pool }P^A=\sum_k\pi_kP_k.
 \label{eq:intro_reverse_pool}
\end{align}
The first requires suitable common support and integrability; the second holds for arbitrary probability laws \cite{Abbas2009,Csiszar1975,GenestZidek1986}. We call these the \emph{candidate-first} and \emph{expert-first} orders, according to the position of the variable law $Q$.

Both have a financial interpretation. With physical beliefs $P$ fixed and a martingale measure $Q$ variable, logarithmic utility has the dual objective $\KL(P\|Q)$, while exponential utility leads to $\KL(Q\|P)$. Vecer~\cite{VecerNumeraire2026} shows that, on compatible martingale-measure domains, the logarithmic selector commutes with the change-of-numeraire likelihood transform, whereas the minimal-entropy selector need not. Those results concern pricing-measure selection. They motivate retaining both directions here, where the unknown is a consensus law.

The expert-first order also follows from weighted logarithmic wealth. For a pricing numeraire $Y$, pricing measure $P^Y$, and normalized likelihood payoff $Z=\dd Q/\dd P^Y$,
\begin{equation}\label{eq:intro_log_wealth}
 \sum_k\pi_k\E^{P_k}\log Z
 =\sum_k\pi_k\KL(P_k\|P^Y)-\sum_k\pi_k\KL(P_k\|Q)
\end{equation}
under the support and finiteness conditions stated below. The arithmetic pool maximizes this objective over the unrestricted claim class. If $Z_t^k=\E^{P^Y}[\dd P_k/\dd P^Y\mid\F_t]$, then
\begin{equation}\label{eq:intro_dynamic_weights}
 Z_t^A=\sum_k\pi_kZ_t^k,\qquad
 w_t^k=\frac{\pi_kZ_t^k}{Z_t^A}.
\end{equation}
These weights are posterior expert probabilities, or relative subfund wealth shares when the likelihood claims are attainable \cite{Cover1991,HoetingEtAl1999,Vecer2026,VecerRichardTaylor2025}.

Covariance disagreement creates a further difficulty. A continuous path reveals its quadratic variation, and incompatible quadratic-variation identities make diffusion laws mutually singular. The candidate-first path-space entropy objective is then infinite for every candidate. The common-covariance construction of Jaimungal and Pesenti~\cite{JaimungalPesenti2026} cannot be applied directly across these laws. The arithmetic mixture remains defined, but generally is not a diffusion with coefficients depending only on the current time and state. A diffusion-valued aggregate therefore requires an additional approximation criterion.

We measure drift disagreement in the covariance of the second argument, as in Gaussian relative entropy, and covariance disagreement by optimal quadratic matching of Gaussian shocks. A fixed positive definite state metric $M$ specifies the relative importance of state directions. With $d_M$ the resulting covariance distance, the local cost is
\begin{equation}\label{eq:local_cost_intro}
 \ellit_{\lambda,M}((b,A),(\bar b,B))
 =\frac12(b-\bar b)^\top B^{-1}(b-\bar b)
 +\frac\lambda2d_M^2(A,B),\qquad\lambda>0.
\end{equation}
For $M=I_d$, $d_M$ is the Bures--Wasserstein distance. The normalization, metric, and penalty are specified modelling choices. The Gaussian calculation identifies their one-step cost, and the Euler theorem proves convergence of the resulting discretized functional; neither selects those choices uniquely.

The contributions are as follows.
\begin{enumerate}[label=(\roman*),leftmargin=2em]
\item We formulate a directed local-characteristic criterion across heterogeneous covariance specifications, establish linear-coordinate equivariance with the metric transformed accordingly, and prove an $O(h^{1/2})$ approximation bound for its renormalized Euler-kernel functional (\Cref{sec:divergence,sec:approx}).
\item We solve the expert-first coefficient projection. Its drift is the current-state posterior mean, and its covariance balances transport loss against drift dispersion. We compute the alternative first-argument normalization, distinguish this projection from marginal-preserving mimicking, and prove global coefficient regularity for Ornstein--Uhlenbeck experts sharing a mean-reversion rate, with common input volatility on $[0,T]$ and heterogeneous input volatilities on $[t_0,T]$, $t_0>0$ (\Cref{sec:reverse_projection}).
\item We characterize candidate-first aggregation by diffusion control. A congruence formula reduces the unrestricted matrix covariance selector to an ordinary Bures--Wasserstein barycenter, without commutativity assumptions. The condition $H+\lambda M\succ0$ is necessary and sufficient for finiteness of that local problem. We also characterize the constrained selector, which remains finite on a compact covariance set (\Cref{sec:barycenter}).
\item We interpret $\lambda$ through a budget for covariance disagreement and establish the associated monotone trade-off. In the scalar expert-first problem, this yields an explicit calibration to a volatility tolerance (\Cref{subsec:covariance_budget,rem:scalar_tolerance}).
\item We identify distinct effects of drift disagreement: volatility inflation under the expert-first normalization and dynamic volatility reduction in a candidate-first linear--quadratic model. Exact Ornstein--Uhlenbeck transitions reveal a signed correction beyond the Euler criterion; martingale and square-root examples show how financial restrictions alter aggregation (\Cref{sec:LQ_example,sec:OU,sec:finance,sec:cir}).
\end{enumerate}

The economic distinction is where approximation error is evaluated. Fixed expert state distributions lead to posterior-weighted pointwise optimization. Evaluation under the aggregate's own distribution allows it to change future exposure to disagreement and requires dynamic programming. A marginal-preserving diffusion is instead appropriate when the objective is to retain mixture expectations of European payoffs \cite{BrunickShreve2013,Gyongy1986}. A pricing aggregate must also satisfy martingale restrictions for traded assets.

Adapted Wasserstein methods compare processes while preserving their information structure \cite{BackhoffBartlBeiglbockEder2020,BartlBeiglbockPammer2026,BartlBeiglbockPammerSchrottZhang2025}. Related work studies bicausal diffusion transport, multicausal barycenters, and Gaussian adapted geometry \cite{AcciaioBartlGrassHouPammer2026,AcciaioKrsekPammer2025,BackhoffKallbladRobinson2025,GunasingamMattesiniWieselWong2026,GunasingamWong2025,HitzRobinson2024}. Filtered Gaussian transport and an adapted Brenier theorem further describe optimal couplings \cite{BeiglbockPammerSchrott2025,GunasingamWong2026}. These are path-space transport problems. Our general criterion combines local covariance transport with directed drift information; its deterministic Gaussian martingale reduction is a consistency result within that literature.

\Cref{sec:setup,sec:orientations} give the diffusion setup and classical pooling identities. The remaining sections follow the contributions above; the scalar martingale-constrained calculation is proved in the appendix.

\section{Diffusion models and the quadratic-variation obstruction}\label{sec:setup}

Let $\Omega=C([0,T];\R^d)$ with canonical process $X$, canonical filtration $(\F_t)_{0\le t\le T}$, and Borel sigma-field $\F_T$.  A diffusion model is a pair $m=(b,a)$ with
\[
 b:[0,T]\times\R^d\to\R^d,
 \qquad a:[0,T]\times\R^d\to\Spp.
\]
We write $P^m=P^{b,a}$ for the law under which
\begin{equation}\label{eq:sde_general}
 \dd X_t=b(t,X_t)\dd t+a(t,X_t)^{1/2}\dd W_t,
 \qquad X_0=x_0,
\end{equation}
where the principal symmetric square root is used only as a representative; all definitions below depend on $a$, not on a chosen factorization.

\begin{assumption}[Regular diffusion class]\label{ass:regular}
There are constants $0<\underline a\le \overline a<\infty$ and $L<\infty$ such that every model considered satisfies
\[
 \underline a I_d\le a(t,x)\le \overline a I_d,
\]
The coefficients are globally Lipschitz in $x$, uniformly in $t$, and $|b(t,x)|\le L(1+|x|)$. Their time regularity is uniform on compact time intervals, including the endpoints, with the explicit envelope
\begin{equation}\label{eq:time_regularity}
 |b(t,x)-b(s,x)|+\norm{a(t,x)-a(s,x)}_F
 \le L(1+|x|)|t-s|^{1/2},\qquad s,t\in[0,T],\ x\in\R^d.
\end{equation}
The constants are common whenever a class of models is compared. The initial state $x_0$ is common.
\end{assumption}

Assumption \ref{ass:regular} gives a unique strong solution and uniform moment bounds.  More general weakly well-posed martingale problems can be used, but the regular class keeps the coefficient estimates and approximation arguments explicit.

For a continuous path $\omega$, let $[\omega]$ denote its pathwise quadratic variation along dyadic partitions whenever the limit exists.  Under $P^{b,a}$,
\begin{equation}\label{eq:qv_identity}
 [X]_t=\int_0^t a(s,X_s)\dd s,
 \qquad 0\le t\le T,
 \quad P^{b,a}\text{-a.s.}
\end{equation}

\begin{proposition}[Quadratic variation forces singularity]\label{prop:singularity}
Let $m_i=(b_i,a_i)$, $i=1,2$, satisfy Assumption \ref{ass:regular}.  Define
\[
 \Gamma_i:=\left\{\omega:\ [\omega]_t=\int_0^t a_i(s,\omega_s)\dd s
 \text{ for every }t\in[0,T]\right\}.
\]
Then $P^{m_i}(\Gamma_i)=1$.  If
\begin{equation}\label{eq:sing_condition}
 P^{m_2}\left(\int_0^T\norm{a_1(t,X_t)-a_2(t,X_t)}_F\dd t>0\right)=1,
\end{equation}
then $P^{m_2}(\Gamma_1)=0$ and $P^{m_1}\perp P^{m_2}$.  In particular, $\KL(P^{m_1}\|P^{m_2})=+\infty$.
\end{proposition}

\begin{proof}
The identity $P^{m_i}(\Gamma_i)=1$ is \eqref{eq:qv_identity}.  On $\Gamma_1\cap\Gamma_2$ the absolutely continuous matrix-valued functions
\[
 t\mapsto\int_0^t a_1(s,\omega_s)\dd s,
 \qquad
 t\mapsto\int_0^t a_2(s,\omega_s)\dd s
\]
coincide, hence their derivatives coincide for Lebesgue-a.e. $t$.  Therefore
\[
 \Gamma_1\cap\Gamma_2\subseteq
 \left\{\int_0^T\norm{a_1(t,X_t)-a_2(t,X_t)}_F\dd t=0\right\}.
\]
Condition \eqref{eq:sing_condition} gives $P^{m_2}(\Gamma_1)=0$, while $P^{m_1}(\Gamma_1)=1$.  Thus the laws are singular.  Relative entropy is infinite whenever the first measure is not absolutely continuous with respect to the second.
\end{proof}

\begin{remark}
For deterministic covariance rates $a_i(t)$, any difference on a set of positive Lebesgue measure implies singularity.  For a Cox--Ingersoll--Ross variance factor, different coefficients governing volatility of volatility give different quadratic variation $[V]_t=\int_0^t\xi_i^2V_s\dd s$ and hence singular laws under the usual positivity assumptions.  This is the financially relevant case developed in \Cref{sec:cir}.
\end{remark}

\section{Geometric and arithmetic consensus on law space}\label{sec:orientations}

Let $P_1,\ldots,P_K$ be probability measures on a measurable space $(\Omega,\mathcal F)$, with weights $\pi_k>0$ and $\sum_k\pi_k=1$.  The two Kullback--Leibler orientations select different consensus laws.  We recall the classical pooling identities \cite{Abbas2009,Csiszar1975,GenestZidek1986} before relating them to logarithmic wealth and diffusion projection.

\begin{proposition}[Classical Kullback--Leibler pooling identities]\label{thm:two_KL_pools}
Let $\pi_k>0$ and $\sum_k\pi_k=1$.
\begin{enumerate}[label=(\roman*),leftmargin=2em]
\item Suppose $P_k\ll R$ with strictly positive densities $Z_k$ and
$0<c_G:=\E^R\prod_kZ_k^{\pi_k}<\infty$. Define
\begin{equation}\label{eq:geometric_pool}
 \frac{\dd P^G}{\dd R}=c_G^{-1}\prod_kZ_k^{\pi_k}.
\end{equation}
For every $Q\ll R$ for which the terms are well defined,
\begin{equation}\label{eq:forward_KL_decomp}
 \sum_k\pi_k\KL(Q\|P_k)=\KL(Q\|P^G)-\log c_G.
\end{equation}
Thus $P^G$ is the unique candidate-first minimizer.
\item For arbitrary $P_k$, put
\begin{equation}\label{eq:arithmetic_pool}
 P^A=\sum_k\pi_kP_k.
\end{equation}
Then, for every probability law $Q$,
\begin{equation}\label{eq:reverse_KL_decomp}
 \sum_k\pi_k\KL(P_k\|Q)
 =\operatorname{JS}_{\pi}(P_1,\ldots,P_K)+\KL(P^A\|Q),
\end{equation}
where
\begin{equation}\label{eq:JS_def}
 \operatorname{JS}_{\pi}:=\sum_k\pi_k\KL(P_k\|P^A)
 \le-\sum_k\pi_k\log\pi_k<\infty.
\end{equation}
Hence $P^A$ is the unique expert-first minimizer. If $P_k\ll R$, then
\begin{equation}\label{eq:arithmetic_density}
 \dd P^A/\dd R=\sum_k\pi_k Z_k.
\end{equation}
\end{enumerate}
\end{proposition}

\begin{proof}
Expanding $\log(\dd Q/\dd P_k)$ and summing proves (i). For (ii),
$P^A\ge\pi_kP_k$ gives $\KL(P_k\|P^A)\le-\log\pi_k$.
The relative-entropy chain rule and $\sum_k\pi_kP_k=P^A$ yield
\eqref{eq:reverse_KL_decomp}, including infinite values. Both minimizers follow from nonnegativity of relative entropy.
\end{proof}

\subsection{Utility duality and the order of the arguments}
\label{subsec:numeraire_motivation}

For a fixed physical law $P$ and an equivalent martingale measure $Q$, write $Z_Q=\dd Q/\dd P$. The conjugate $V(y)=\sup_x\{U(x)-xy\}$ of logarithmic or exponential utility gives, for $y>0$ and $\gamma>0$,
\begin{align}
 \E^P[V_{\log}(yZ_Q)]
 &=-\log y-1+\KL(P\|Q),\label{eq:numeraire_log_dual}\\*
 \E^P[V_{\exp}(yZ_Q)]
 &=\frac y\gamma\left\{\KL(Q\|P)+\log\frac y\gamma-1\right\},
 \label{eq:numeraire_exp_dual}
\end{align}
where $U_{\log}(x)=\log x$ and $U_{\exp}(x)=-e^{-\gamma x}$. Thus logarithmic utility leads to the expert-first order and the minimal-entropy martingale measure to the candidate-first order \cite{Frittelli2000}. The latter is also used in entropy-penalized robust control \cite{HansenSargent2008}.

Vecer~\cite{VecerNumeraire2026} relates these objectives to numeraire choice. For positive traded numeraires $N,Y$, let $L=(Y_T/N_T)/(Y_0/N_0)$. On a martingale-measure class satisfying $\E^Q L=1$, the likelihood transform is $T_LQ=LQ$. For fixed $P$ and $\E^P|\log L|<\infty$,
\begin{equation}\label{eq:numeraire_constant_shift}
 \KL(P\|T_LQ)=\KL(P\|Q)-\E^P\log L.
\end{equation}
When $T_L$ maps the admissible classes bijectively, the minimizers of the logarithmic objective are therefore carried into one another. The change in $\KL(Q\|P)$ generally depends on $Q$, so the analogous conclusion need not hold.

These are statements about constrained pricing-measure selection with physical beliefs held fixed. Both unrestricted pools in \Cref{thm:two_KL_pools} are independent of the dominating measure used to write their densities. Neither fact makes the covariance-transport loss introduced below invariant under a change of numeraire. The present aggregation problems retain the two argument orders while specifying that additional loss separately.

\subsection{Likelihood payoffs and arithmetic aggregation}

\begin{corollary}[Weighted logarithmic wealth selects the arithmetic pool]
\label{cor:log_wealth_arithmetic}
Fix a pricing numeraire $Y$ with pricing measure $P^Y$. Suppose
$P_k\ll P^Y$, $\KL(P_k\|P^Y)<\infty$ for every $k$, and
$P^A=\sum_k\pi_kP_k\sim P^Y$. Write
$Z_k=\dd P_k/\dd P^Y$. For every strictly positive normalized claim
$Z=\dd Q/\dd P^Y$, with $\E^{P^Y}Z=1$,
\begin{equation}\label{eq:weighted_log_identity}
 \sum_{k=1}^K\pi_k\E^{P_k}[\log Z]
 =\sum_{k=1}^K\pi_k\KL(P_k\|P^Y)
  -\sum_{k=1}^K\pi_k\KL(P_k\|Q).
\end{equation}
The identity holds in $[-\infty,\infty)$. Over all such claims, the
unique maximizer is the arithmetic likelihood payoff
\begin{equation}\label{eq:arithmetic_payoff}
 Z^A=\frac{\dd P^A}{\dd P^Y}=\sum_{k=1}^K\pi_kZ_k.
\end{equation}
Under the hypotheses of \Cref{thm:two_KL_pools}(i), with $R=P^Y$,
the geometric pool instead has normalized payoff
\begin{equation}\label{eq:geometric_payoff}
 Z^G=\frac{\prod_k Z_k^{\pi_k}}
 {\E^{P^Y}[\prod_kZ_k^{\pi_k}]}.
\end{equation}
\end{corollary}

\begin{proof}
Finite $\KL(P_k\|P^Y)$ and $\E^{P^Y}Z=1$ imply
$\E^{P_k}(\log Z)^+\le\KL(P_k\|P^Y)+\log 2$, by the entropy
inequality applied to $\log\max(1,Z)$. Hence each expected logarithm
is well defined, and the Radon--Nikodym chain rule gives
\eqref{eq:weighted_log_identity}. The optimization now follows from
\Cref{thm:two_KL_pools}(ii). Equivalence $P^A\sim P^Y$ makes the
maximizing claim strictly positive.
\end{proof}

For the dynamic interpretation, let
\[
 Z_t^k:=\E^{P^Y}[Z_k\mid\mathcal F_t],
 \qquad
 Z_t^A:=\sum_{k=1}^K\pi_kZ_t^k,
 \qquad
 w_t^k:=\frac{\pi_kZ_t^k}{Z_t^A}.
\]
If the expert likelihood claims are attainable, $Z_t^A$ is the value of their aggregate subfund and $w_t^k$ is its relative wealth share in normalized $Y$-units. Under the latent-expert representation, $w_t^k$ is the Bayesian posterior probability of expert $k$.

\begin{proposition}[Dynamic arithmetic pool under common covariance]\label{thm:dynamic_arithmetic_pool}
Let $R=P^{b_0,a}$ be a reference diffusion law and suppose every expert law $P_k=P^{b_k,a}$ has the same covariance field $a$, is equivalent to $R$, and satisfies the Girsanov integrability conditions.  Put
\[
 \theta_k(t,x):=a(t,x)^{-1/2}(b_k(t,x)-b_0(t,x)).
\]
Let $Z_t^k=\E^R[\dd P_k/\dd R\mid\mathcal F_t]$, $Z_t^A=\sum_k\pi_kZ_t^k$, and $w_t^k=\pi_kZ_t^k/Z_t^A$.  Then $Z_T^A=\dd P^A/\dd R$ and, under $P^A$,
\begin{equation}\label{eq:dynamic_arithmetic_drift}
 \dd X_t=\left(\sum_{k=1}^K w_t^k b_k(t,X_t)\right)\dd t
 +a(t,X_t)^{1/2}\dd W_t^A.
\end{equation}
If a latent index $K$ is sampled with $\Pp(K=k)=\pi_k$ and $X\mid\{K=k\}\sim P_k$, then
\begin{equation}\label{eq:posterior_weights}
 w_t^k=P^A(K=k\mid\mathcal F_t).
\end{equation}
Hence the arithmetic pool updates expert influence by Bayesian evidence and, when the likelihood claims are attainable, by relative subfund performance.  In general, the weights depend on the full observed path, so the arithmetic pool need not be Markov in $X_t$ alone.
\end{proposition}

\begin{proof}
Under $R$, $\dd Z_t^k=Z_t^k\theta_k^\top\dd W_t^R$, so
$\dd Z_t^A=Z_t^A(\sum_kw_t^k\theta_k)^\top\dd W_t^R$.
Girsanov gives \eqref{eq:dynamic_arithmetic_drift}, and Bayes' formula gives \eqref{eq:posterior_weights}.
\end{proof}

\begin{proposition}[Arithmetic consensus survives singularity but may leave the Markov diffusion class]\label{prop:mixture_not_markov}
The arithmetic pool $P^A=\sum_k\pi_kP_k$ is well defined for arbitrary expert laws and can always be realized on the enlarged state space $\{1,\ldots,K\}\times\Omega$ by sampling a latent expert index once at time zero.  However, it need not equal $P^{b,a}$ for any Markov coefficients $(b(t,x),a(t,x))$.

In particular, let $P_k$ be the law of $X_t=\sigma_kW_t$ on $C([0,T];\R)$, with $\sigma_1\ne\sigma_2$ and positive weights.  Then $P^A=\pi_1P_1+\pi_2P_2$ is not the law of an It\^o diffusion
\[
 \dd X_t=b(t,X_t)\dd t+\sqrt{a(t,X_t)}\dd W_t
\]
with Borel coefficients depending only on $(t,X_t)$.
\end{proposition}

\begin{proof}
The latent-index realization is immediate.  For the Brownian example, suppose such a Markov representation existed.  Pathwise quadratic variation would imply
\[
 a(t,X_t)=\sigma_k^2
 \quad \dd t\otimes P_k\text{-a.e.},
 \qquad k=1,2,
\]
because $P_k\ll P^A$.  For Lebesgue-a.e. $t>0$, the law of $X_t$ under each $P_k$ has a strictly positive Gaussian density on all of $\R$.  Hence $a(t,x)=\sigma_1^2$ for Lebesgue-a.e. $x$ and also $a(t,x)=\sigma_2^2$ for Lebesgue-a.e. $x$, a contradiction.
\end{proof}

Thus arithmetic aggregation survives singularity, but a Markov diffusion output requires a further projection.

\section{An information--transport divergence for drift and covariance}\label{sec:divergence}

\subsection{Gaussian covariance transport and drift normalization}

For $A,B\in\Sp$, define the squared Bures--Wasserstein distance by
\begin{equation}\label{eq:bures}
 \BW^2(A,B)
 :=\tr(A)+\tr(B)-2\tr\!\left[(B^{1/2}AB^{1/2})^{1/2}\right].
\end{equation}
It is the squared $2$-Wasserstein distance between centered Gaussian laws $N(0,A)$ and $N(0,B)$.  Equivalently,
\begin{equation}\label{eq:procrustes}
 \BW^2(A,B)=\min_{O\in\mathbb O(d)}\norm{A^{1/2}-B^{1/2}O}_F^2.
\end{equation}
The distance is continuous and jointly convex; it is independent of the orthogonal representation of Brownian shocks \cite{AguehCarlier2011,Gelbrich1990}.

Fix $M\succ0$, common to all experts and candidates, and put $|z|_M^2=z^\top Mz$. The covariance loss is
\begin{equation}\label{eq:metric_bures}
 d_M^2(A,B):=
 \inf_{\substack{\xi\sim N(0,A)\\\zeta\sim N(0,B)}}\E|\xi-\zeta|_M^2
 =\BW^2(M^{1/2}AM^{1/2},M^{1/2}BM^{1/2}).
\end{equation}
Applying $M^{1/2}$ to each coupled shock proves the equality. This generalized Gaussian geometry is studied in \cite{HanEtAl2023}, with parameter $M^{-1}$ in their notation. Here the metric is fixed before aggregation. The next observation records the additional drift-normalization assumption.

\begin{proposition}[Local cost under specified loss principles]
\label{prop:local_cost_principles}
Fix $\lambda>0$ and $M\succ0$. Suppose a finite local cost
$c((b,A),(\bar b,B))$ satisfies:
\begin{enumerate}[label=(\roman*),leftmargin=2em]
\item for $A=B$, it equals $\tfrac12(b-\bar b)^\top B^{-1}(b-\bar b)$;
\item for $b=\bar b$, it equals $\tfrac\lambda2d_M^2(A,B)$;
\item its increment when $b$ replaces $\bar b$ depends only on $b-\bar b$ and $B$, independently of $A$.
\end{enumerate}
Then, and conversely,
\begin{equation}\label{eq:local_cost_metric}
 c((b,A),(\bar b,B))=
 \frac12(b-\bar b)^\top B^{-1}(b-\bar b)+\frac\lambda2d_M^2(A,B).
\end{equation}
\end{proposition}

\begin{proof}
By (iii), evaluate the drift increment at $A=B$ and use (i). Adding the equal-drift cost in (ii) proves the formula; substitution proves the converse.
\end{proof}

Condition (iii) is a modelling assumption. Replacing $B^{-1}$ by
$(1-\eta)B^{-1}+\eta A^{-1}$, $0<\eta\le1$, preserves (i)--(ii).
We retain the second-argument covariance because it is the exact normalization of the Gaussian KL mean term. Its consequences for the expert-first selector, including the alternative $A^{-1}$ convention, are computed in \Cref{rem:normalization_dependence}.

\subsection{A two-scale Gaussian identity}

Let $W_{2,M}^2(\mu,\nu)=\inf\E|U-V|_M^2$, where the infimum is over couplings of $\mu$ and $\nu$. For $b\in\R^d$, $A\in\Spp$, and $h>0$, write
\[
 G_h^{b,A}:=N(hb,hA).
\]
Define the Gaussian covariance contribution
\begin{equation}\label{eq:kappa_cov}
 \mathcal K(A\|B)
 :=\frac12\left\{\tr(B^{-1}A)-d-\log\det(B^{-1}A)\right\}.
\end{equation}

\begin{theorem}[Two-scale local Gaussian decomposition]\label{thm:local_gaussian}
For every $b,\bar b\in\R^d$, $A,B\in\Spp$, $h>0$, and $\lambda>0$, with fixed $M\succ0$,
\begin{align}
 &\frac{1}{h}\left\{
 \KL\!\left(G_h^{b,A}\middle\|G_h^{\bar b,B}\right)-\mathcal K(A\|B)
 \right\}
 +\frac{\lambda}{2h}W_{2,M}^2\!\left(G_h^{0,A},G_h^{0,B}\right)
 \notag\\
 &\hspace{3cm}
 =\frac12(b-\bar b)^\top B^{-1}(b-\bar b)
 +\frac{\lambda}{2}d_M^2(A,B).
 \label{eq:local_gaussian_identity}
\end{align}
In particular, the right-hand side is independent of $h$.
\end{theorem}

\begin{proof}
The Kullback--Leibler divergence between Gaussian distributions is
\[
 \KL\!\left(N(m,C)\middle\|N(\bar m,\bar C)\right)
 =\frac12\left\{\tr(\bar C^{-1}C)+(m-\bar m)^\top\bar C^{-1}(m-\bar m)-d
 +\log\frac{\det\bar C}{\det C}\right\}.
\]
Substituting $(m,C)=(hb,hA)$ and $(\bar m,\bar C)=(h\bar b,hB)$ gives
\[
 \KL(G_h^{b,A}\|G_h^{\bar b,B})
 =\mathcal K(A\|B)+\frac h2(b-\bar b)^\top B^{-1}(b-\bar b).
\]
Applying $M^{1/2}$ to both shocks and using homogeneity gives
$W_{2,M}^2(N(0,hA),N(0,hB))=hd_M^2(A,B)$.
\end{proof}

The theorem displays the two time scales.  If $A\neq B$, each short-time conditional Kullback--Leibler contributes the order-one quantity $\mathcal K(A\|B)$, so the path-space Kullback--Leibler of a grid with $T/h$ increments diverges at order $h^{-1}$.  After this singular covariance entropy is separated, the remaining information term is order $h$ and measures drift.  The centered quadratic transport discrepancy is also order $h$. Its addition retains covariance disagreement after the order-one entropy term has been removed. The identity verifies the specified local cost; it does not select the transport metric or the penalty parameter.

\subsection{The integrated criterion and its two orientations}

\begin{definition}[Information--transport divergence]\label{def:ITD}
Let $m=(b,a)$ and $\bar m=(\bar b,\bar a)$ satisfy Assumption \ref{ass:regular}.  For $\lambda>0$ and fixed $M\succ0$, define
\begin{equation}\label{eq:ITD}
 \D_{\lambda,M}(P^m\|P^{\bar m})
 :=\E^{P^m}\!\left[\int_0^T
 \ellit_{\lambda,M}\big(m(t,X_t),\bar m(t,X_t)\big)\dd t\right],
\end{equation}
where
\begin{equation}\label{eq:ell_def}
 \ellit_{\lambda,M}\big((b,A),(\bar b,B)\big)
 :=\frac12(b-\bar b)^\top B^{-1}(b-\bar b)
 +\frac{\lambda}{2}d_M^2(A,B).
\end{equation}
\end{definition}

The quantity $\D_{\lambda,M}$ is directed and is not a metric.  Its volatility component is symmetric pointwise, but the expectation and drift normalization are taken under the first model and relative to the second model's covariance.

For expert models $m_1,\ldots,m_K$, the same information--transport divergence induces two aggregation objectives:
\begin{align}
 J_{\lambda,M}^{\rightarrow}(m)&:=\sum_{k=1}^K\pi_k\D_{\lambda,M}(P^m\|P^{m_k}),
 \label{eq:forward_IT_objective}\\*
 J_{\lambda,M}^{\leftarrow}(m)&:=\sum_{k=1}^K\pi_k\D_{\lambda,M}(P^{m_k}\|P^m).
 \label{eq:reverse_IT_objective}
\end{align}
The candidate-first objective extends geometric Kullback--Leibler aggregation within the diffusion class.  The expert-first objective projects arithmetic likelihood aggregation into the diffusion class.  Their optimization theories differ because the expectation in \eqref{eq:forward_IT_objective} is taken under the controlled candidate law, whereas the expectations in \eqref{eq:reverse_IT_objective} are taken under the fixed expert laws.

\begin{proposition}[Basic properties]\label{prop:basic}
Under Assumption \ref{ass:regular}:
\begin{enumerate}[label=(\roman*),leftmargin=2em]
\item $\D_{\lambda,M}(P^m\|P^{\bar m})\ge0$.  It vanishes if and only if
$b(t,X_t)=\bar b(t,X_t)$ and $a(t,X_t)=\bar a(t,X_t)$ for $\dd t\otimes P^m$-a.e. $(t,\omega)$.
\item For fixed $\bar m$, the local cost is convex separately in the candidate drift $b$ and covariance $a$; on $\R^d\times\Spp$ it is strictly convex in each variable.
\item On every coefficient class satisfying common uniform ellipticity and growth bounds, $\D_{\lambda,M}$ is finite and locally bounded by a quadratic function of the coefficient differences.
\end{enumerate}
\end{proposition}

\begin{proof}
Nonnegativity follows from positive definiteness of $B^{-1}$ and the quadratic transport representation. Strict convexity of $d_M^2(\cdot,B)$ follows from that of $\BW^2(\cdot,M^{1/2}BM^{1/2})$ by an invertible congruence.  Vanishing of the integral implies vanishing of both nonnegative integrands $\dd t\otimes P^m$-a.e.  For fixed $(\bar b,B)$, the drift term is a strictly convex quadratic function of $b$, while $A\mapsto d_M^2(A,B)$ is convex, and strictly convex on $\Spp$ when $B\succ0$.  This proves (ii) without asserting joint convexity of the full model-to-model functional.  Uniform ellipticity and the linear-growth bounds give (iii).
\end{proof}

\begin{proposition}[Equivariance under linear changes of state]
\label{prop:metric_equivariance}
Let $L$ be invertible and transform the state by $x'=Lx$.  Set
\[
 b'=Lb,\quad \bar b'=L\bar b,\quad
 A'=LAL^\top,\quad B'=LBL^\top,\quad
 M'=L^{-\top}ML^{-1}.
\]
Then
\begin{equation}\label{eq:metric_equivariance}
 d_{M'}^2(A',B')=d_M^2(A,B),\qquad
 \ellit_{\lambda,M'}((b',A'),(\bar b',B'))
 =\ellit_{\lambda,M}((b,A),(\bar b,B)).
\end{equation}
Consequently the integrated divergences and both aggregation problems
are unchanged after transforming the initial state, coefficient fields,
and admissible classes by the same map.
\end{proposition}

\begin{proof}
The map $(\xi,\zeta)\mapsto(L\xi,L\zeta)$ is a bijection between
the coupling classes in \eqref{eq:metric_bures}, and
$|L(\xi-\zeta)|_{M'}^2=|\xi-\zeta|_M^2$.  The drift equality follows
from $(LBL^\top)^{-1}=L^{-\top}B^{-1}L^{-1}$.  Integration under the
pushforward laws and the bijection of admissible models prove the last
claim.
\end{proof}

Changing measurement units therefore preserves the aggregate when the
loss metric is transformed with them.  Resetting $M$ to the identity
after a general linear change specifies a different covariance loss.
This result does not assert invariance under nonlinear state changes.
For a $C^2$ diffeomorphism $f$, the transformed drift difference is
\[
 Df\,(b-\bar b)+\tfrac12 D^2f:(A-B),
\]
where the contraction is componentwise.  Thus price and log-price
coordinates can lead to different criteria when the covariances differ.

\subsection{Metric normalization and a covariance-disagreement budget}\label{subsec:covariance_budget}

The scale of the state metric must be fixed before interpreting $\lambda$.
For $c>0$,
\[
 d_{cM}^2(A,B)=c\,d_M^2(A,B),
 \qquad
 \lambda d_{cM}^2(A,B)=(c\lambda)d_M^2(A,B).
\]
Thus the scale of $M$ and the value of $\lambda$ cannot be identified
separately from the objective.  One convention is to fix a reference
covariance of state changes $\Sigma_{\mathrm{ref}}\succ0$ and take
$M=\Sigma_{\mathrm{ref}}^{-1}$.  If the reference input is a covariance
\emph{rate} $A_{\mathrm{ref}}$, a reference horizon
$\tau_{\mathrm{ref}}>0$ specifies
$\Sigma_{\mathrm{ref}}=\tau_{\mathrm{ref}}A_{\mathrm{ref}}$.
State changes are then measured in reference-standard-deviation units;
both integrated components of the criterion, and hence $\lambda$, are
dimensionless.  This normalization is a modeling convention, fixed in
advance and common to all experts and candidates.  The Euler mesh is an
approximation parameter and does not determine $\lambda$.

For $m=(b,a)$, separate the two aggregation objectives as
\begin{equation}\label{eq:budget_decomposition}
 J_{\lambda,M}^{\diamond}(m)
 =\mathcal I^{\diamond}(m)+\lambda\mathcal T_M^{\diamond}(m),
 \qquad \diamond\in\{\rightarrow,\leftarrow\},
\end{equation}
where
\begin{align*}
 \mathcal I^{\rightarrow}(m)
 &=\frac12\sum_k\pi_k\E^{P^m}\int_0^T
 (b-b_k)^\top a_k^{-1}(b-b_k)(t,X_t)\dd t,\\
 \mathcal T_M^{\rightarrow}(m)
 &=\frac12\sum_k\pi_k\E^{P^m}\int_0^T
 d_M^2\big(a(t,X_t),a_k(t,X_t)\big)\dd t,\\
 \mathcal I^{\leftarrow}(m)
 &=\frac12\sum_k\pi_k\E^{P^{m_k}}\int_0^T
 (b_k-b)^\top a^{-1}(b_k-b)(t,X_t)\dd t,\\
 \mathcal T_M^{\leftarrow}(m)
 &=\frac12\sum_k\pi_k\E^{P^{m_k}}\int_0^T
 d_M^2\big(a_k(t,X_t),a(t,X_t)\big)\dd t.
\end{align*}
In either orientation, a covariance-disagreement budget gives the problem
\begin{equation}\label{eq:covariance_budget}
 \inf_{m\in\mathfrak M}\mathcal I^{\diamond}(m)
 \quad\text{subject to}\quad
 \mathcal T_M^{\diamond}(m)\le\delta,
\end{equation}
on a prescribed admissible class $\mathfrak M$.  The following elementary result relates the budget to the penalized objective.

\begin{proposition}[Covariance budgets and the information--transport trade-off]
\label{prop:budget_tradeoff}
Fix either orientation and suppress its superscript.  Let $M$ and
$\mathfrak M$ be independent of $\lambda$.  Suppose that some
$m\in\mathfrak M$ has both $\mathcal I(m)<\infty$ and
$\mathcal T_M(m)<\infty$, and define
\[
 F(\lambda):=\inf_{m\in\mathfrak M}
 \{\mathcal I(m)+\lambda\mathcal T_M(m)\},\qquad \lambda>0.
\]
Then $F$ is finite, nondecreasing, and concave on $(0,\infty)$.
If $m_\lambda$ is a global minimizer, it also globally minimizes
\eqref{eq:covariance_budget} for
$\delta=\mathcal T_M(m_\lambda)$.
If global minimizers $m_{\lambda_1}$ and $m_{\lambda_2}$ exist for
$0<\lambda_1<\lambda_2$, write
$\mathcal I_i=\mathcal I(m_{\lambda_i})$ and
$\mathcal T_i=\mathcal T_M(m_{\lambda_i})$.  Then
\begin{equation}\label{eq:budget_monotonicity}
 \mathcal T_2\le\mathcal T_1,
 \qquad
 \lambda_1(\mathcal T_1-\mathcal T_2)
 \le\mathcal I_2-\mathcal I_1
 \le\lambda_2(\mathcal T_1-\mathcal T_2).
\end{equation}
At any $\lambda>0$ where $F$ is differentiable and a global minimizer
exists,
\begin{equation}\label{eq:budget_envelope}
 F'(\lambda)=\mathcal T_M(m_\lambda).
\end{equation}
\end{proposition}

\begin{proof}
Nonnegativity of both components and the finite competitor give
$0\le F(\lambda)<\infty$.  The infimum of affine nondecreasing functions
of $\lambda$ is concave and nondecreasing.  Any model with covariance
discrepancy at most $\mathcal T_M(m_\lambda)$ and strictly smaller
information cost would have a strictly smaller penalized objective,
which proves the budget assertion.  Optimality at the two parameters
gives
\[
 \mathcal I_1+\lambda_1\mathcal T_1
 \le\mathcal I_2+\lambda_1\mathcal T_2,
 \qquad
 \mathcal I_2+\lambda_2\mathcal T_2
 \le\mathcal I_1+\lambda_2\mathcal T_1.
\]
Adding yields $\mathcal T_2\le\mathcal T_1$; rearranging gives the other
bounds in \eqref{eq:budget_monotonicity}.  Finally, for every $\mu>0$,
\[
 F(\mu)\le F(\lambda)
 +(\mu-\lambda)\mathcal T_M(m_\lambda).
\]
The difference quotients from either side give
\eqref{eq:budget_envelope} whenever $F$ is differentiable.
\end{proof}

The proposition needs no convexity of the model class. Its monotonicity concerns integrated covariance discrepancy. A converse for every prescribed budget requires additional duality assumptions.

More precisely, let $v(\delta)$ denote the value of
\eqref{eq:covariance_budget}, with value $+\infty$ for an infeasible
budget.  For every $\lambda>0$,
\[
 v(\delta)\ge F(\lambda)-\lambda\delta,
\]
with equality at $\delta_\lambda=\mathcal T_M(m_\lambda)$ whenever a
global penalized minimizer exists.  Thus $-\lambda$ is a supporting slope
of the budget value at $\delta_\lambda$.  If $v$ is finite in a
neighborhood of that budget and differentiable there, then
$v'(\delta_\lambda)=-\lambda$: the multiplier is the marginal reduction
in optimal drift-information cost from relaxing the covariance budget.

\subsection{Exact reductions and the coordinate convention}

For the process results below, work in coordinates $y=M^{1/2}x$ and suppress the change in notation. By \Cref{prop:metric_equivariance}, the loss metric is then the identity. We write $\D_\lambda=\D_{\lambda,I}$, $\ellit_\lambda=\ellit_{\lambda,I}$, and $J^{\diamond}=J_{\lambda,I}^{\diamond}$, with $\diamond\in\{\rightarrow,\leftarrow\}$. Coefficients, initial states, and admissible sets are all expressed in these coordinates. This convention covers every fixed positive definite metric. The matrix selector in \Cref{thm:congruence_selector} is also stated for general $M$, so it can be applied directly in the original coordinates. The worked examples choose $M=I$ in the displayed state variables.

\begin{theorem}[Reduction to path-space Kullback--Leibler divergence]\label{thm:KL_reduction}
Suppose $m=(b,a)$ and $\bar m=(\bar b,a)$ have the same covariance field.  Assume the Girsanov exponential is a true martingale, for example Novikov's condition holds.  Then
\begin{equation}\label{eq:KL_reduction}
 \D_\lambda(P^{b,a}\|P^{\bar b,a})
 =\KL(P^{b,a}\|P^{\bar b,a})
 =\frac12\E^{P^{b,a}}\int_0^T
 (b-\bar b)^\top a^{-1}(b-\bar b)(t,X_t)\dd t.
\end{equation}
\end{theorem}

\begin{proof}
The Bures--Wasserstein term is zero.  The remaining equality is the standard entropy formula obtained from Girsanov's theorem; see, e.g., \cite{KaratzasShreve1991}.
\end{proof}

Thus $\D_\lambda$ is a genuine extension of relative entropy on each fixed-volatility absolute-continuity class.  Which consensus it induces depends on whether the candidate occupies the first or second argument.

To state the transport case, let $\AW_{2,T}$ denote the bicausal transport cost with terminal quadratic loss,
\[
 \AW_{2,T}^2(P,Q):=\inf_{\pi\in\Pi_{bc}(P,Q)}\E^\pi|X_T-Y_T|^2,
\]
where $\Pi_{bc}$ is the set of bicausal couplings.  This terminal-loss cost is used only for the exact deterministic-Gaussian reduction below.  In \Cref{sec:approx} we use the distinct path-$L^2$ adapted cost $\AW_{2,L^2}$ to control Euler approximations over the whole trajectory.  The next result is included as a consistency statement, not as a claim to the Gaussian transport geometry developed in \cite{AcciaioBartlGrassHouPammer2026,GunasingamMattesiniWieselWong2026,GunasingamWong2025}.

\begin{theorem}[Deterministic Gaussian adapted-transport case]\label{thm:AW_reduction}
Let $b=\bar b=0$ and let $a(t),\bar a(t)$ be deterministic, measurable, uniformly elliptic covariance rates.  Then
\begin{equation}\label{eq:AW_reduction}
 \AW_{2,T}^2(P^{0,a},P^{0,\bar a})
 =\int_0^T\BW^2(a(t),\bar a(t))\dd t,
\end{equation}
and consequently
\begin{equation}\label{eq:D_AW}
 \D_\lambda(P^{0,a}\|P^{0,\bar a})
 =\frac\lambda2\AW_{2,T}^2(P^{0,a},P^{0,\bar a}).
\end{equation}
\end{theorem}

\begin{proof}
Let $\pi$ be bicausal and let $\mathbb G$ be the joint filtration.  Causality from $X$ to $Y$ implies that the past of $Y$ is conditionally independent of the future of $X$ given the past of $X$; hence the $P^{0,a}$-martingale $X$ remains a martingale in $\mathbb G$.  Reverse causality gives the same statement for $Y$.  Thus $(X,Y)$ is a continuous $\mathbb G$-martingale.  Write its instantaneous cross-covariance as $c_t\dd t=\dd\langle X,Y\rangle_t$.  Positive semidefiniteness of the block covariance matrix implies
$c_t=a(t)^{1/2}K_t\bar a(t)^{1/2}$ for a predictable contraction $K_t$.  It\^o isometry gives
\[
 \E^\pi|X_T-Y_T|^2
 =\int_0^T\{\tr a(t)+\tr\bar a(t)-2\tr c_t\}\dd t.
\]
The maximal trace over contractions is
$\tr[(\bar a(t)^{1/2}a(t)\bar a(t)^{1/2})^{1/2}]$, which yields the lower bound in \eqref{eq:AW_reduction}.  Choose a measurable orthogonal Procrustes optimizer $O(t)$ and, on the law of $X$, define
\begin{equation}\label{eq:explicit_gaussian_coupling}
 Y_t:=x_0+\int_0^t \bar a(s)^{1/2}O(s)a(s)^{-1/2}\dd X_s.
\end{equation}
Then $Y$ has covariance rate $\bar a$, while
$\dd\langle X,Y\rangle_t=a(t)^{1/2}O(t)^\top\bar a(t)^{1/2}\dd t$, and the chosen $O(t)$ attains the maximal trace pointwise.  Because the deterministic integrands in \eqref{eq:explicit_gaussian_coupling} are invertible, each coordinate filtration reconstructs the common driving Brownian filtration.  The coupling is therefore bicausal and attains the lower bound; see also the synchronous-coupling results for SDE laws in \cite{BackhoffKallbladRobinson2025,HitzRobinson2024}.
\end{proof}

\begin{remark}[Relation to Gaussian adapted barycenters]
In discrete time, Gaussian adapted-Wasserstein barycenters reduce to classical Bures--Wasserstein barycenter problems for suitable covariance columns \cite{GunasingamMattesiniWieselWong2026}.  Equation \eqref{eq:AW_reduction} is the deterministic-rate continuous-time case that the proposed divergence recovers.  The remaining problem is the interior control problem in which both drift and covariance are optimized jointly.
\end{remark}

\begin{remark}[Relation to a specific Wasserstein divergence]
In one dimension, \cite[Theorem 2.9]{BackhoffZhang2026} gives, under its hypotheses,
\[
 \SW_2(P^m\|P^{\bar m})
 =\frac{2}{\pi}\E^{P^m}\int_0^T
 \left(|\sigma(t,X_t)|-|\bar\sigma(t,X_t)|\right)^2\dd t
\]
for the martingale components.  Hence the volatility part of \eqref{eq:ITD} equals
$(\lambda\pi/4)\SW_2$ in that scalar setting.  The Bures--Wasserstein term is the multidimensional local-covariance counterpart of this specific Wasserstein divergence.
\end{remark}

\section{Approximation and adapted stability}\label{sec:approx}

Let $\pi_h=\{t_j=jh:j=0,\ldots,N\}$, $h=T/N$.  For a model $m=(b,a)$, let $X^h$ be its continuous-time Euler interpolation,
\begin{equation}\label{eq:euler}
 X^h_t=x_0+\int_0^t b(\eta_h(s),X^h_{\eta_h(s)})\dd s
 +\int_0^t a(\eta_h(s),X^h_{\eta_h(s)})^{1/2}\dd W_s,
\end{equation}
where $\eta_h(s)=t_j$ for $s\in[t_j,t_{j+1})$.

The one-step Euler kernels are Gaussian.  For two models $m,\bar m$, define the discrete local divergence
\begin{align}\label{eq:discrete_div}
 \D_{\lambda,h}(m\|\bar m)
 :=\E\sum_{j=0}^{N-1}h\,
 \ellit_\lambda\big(
 m(t_j,X^h_{t_j}),\bar m(t_j,X^h_{t_j})
 \big).
\end{align}
For $x\in\R^d$, let
\[
 K_{j,h}^{m}(x):=N\!\left(x+h b(t_j,x),h a(t_j,x)\right),
 \qquad
 \widehat K_{j,h}^{m}(x):=N\!\left(0,h a(t_j,x)\right).
\]
Then \Cref{thm:local_gaussian} gives the exact identity
\begin{align}\label{eq:euler_exact_identity}
 \D_{\lambda,h}(m\|\bar m)
 =\E\sum_{j=0}^{N-1}\Big[&
 \KL\!\left(K_{j,h}^{m}(X^h_{t_j})\middle\|K_{j,h}^{\bar m}(X^h_{t_j})\right)
 -\mathcal K\!\left(a(t_j,X^h_{t_j})\middle\|\bar a(t_j,X^h_{t_j})\right)\\
 &+\frac\lambda2W_2^2\!\left(
 \widehat K_{j,h}^{m}(X^h_{t_j}),
 \widehat K_{j,h}^{\bar m}(X^h_{t_j})\right)\Big].\notag
\end{align}
Thus $\D_{\lambda,h}$ is exactly the displayed renormalized Euler-kernel functional.  The passage $h\downarrow0$ concerns only the Euler state process and the Riemann sum.  Exact transition kernels are treated separately in \Cref{sec:OU}.

\begin{theorem}[Euler convergence]\label{thm:euler_conv}
Under Assumption \ref{ass:regular}, there is a constant $C$, depending only on the common bounds, $T$, and fixed $\lambda$, such that
\begin{equation}\label{eq:euler_div_rate}
 \left|\D_{\lambda,h}(m\|\bar m)-\D_\lambda(P^m\|P^{\bar m})\right|
 \le C(1+|x_0|^3)h^{1/2}.
\end{equation}
Moreover, if $\AW_{2,L^2}$ is the adapted Wasserstein distance associated with the path cost $\int_0^T|x_t-y_t|^2\dd t$, then
\begin{equation}\label{eq:euler_aw_rate}
 \AW_{2,L^2}(\law(X^h),P^m)\le C(1+|x_0|)h^{1/2}.
\end{equation}
\end{theorem}

\begin{proof}
The synchronous coupling of $X^h$ and $X$ is bicausal.  Standard strong Euler estimates under global Lipschitz conditions give
\[
 \E\sup_{t\le T}|X_t^h-X_t|^2\le C(1+|x_0|^2)h.
\]
This immediately yields \eqref{eq:euler_aw_rate}.  Uniform ellipticity and Lipschitz continuity imply that the map
$x\mapsto\ellit_\lambda(m(t,x),\bar m(t,x))$ is locally Lipschitz with at most quadratic growth, uniformly in $t$.  More precisely, its spatial increment is bounded by $C(1+|x|^2+|y|^2)|x-y|$, and its time increment by $C(1+|x|^3)|t-s|^{1/2}$. Cauchy--Schwarz, the fourth-moment bounds, and the strong Euler estimate therefore give the Riemann-sum estimate \eqref{eq:euler_div_rate}.
\end{proof}

The next statement is useful for parametric stochastic-volatility families and for numerical model aggregation.

\begin{corollary}[Stability of candidate-first parametric barycenters]\label{cor:parametric}
Let $\Theta$ be compact and let $m_\theta=(b_\theta,a_\theta)$ be a family satisfying Assumption \ref{ass:regular} uniformly in $\theta$, with coefficients continuous in $\theta$ locally uniformly in $(t,x)$.  Fix expert models $m_1,\ldots,m_K$ and weights $\pi_k>0$.  Define
\[
 J^{\rightarrow}(\theta)=\sum_{k=1}^K\pi_k\D_\lambda(P^{m_\theta}\|P^{m_k}),
 \qquad
 J_h^{\rightarrow}(\theta)=\sum_{k=1}^K\pi_k\D_{\lambda,h}(m_\theta\|m_k).
\]
Then $J_h^{\rightarrow}\to J^{\rightarrow}$ uniformly on $\Theta$ at rate $O(h^{1/2})$.  Hence
\[
 \min_{\theta\in\Theta}J_h^{\rightarrow}(\theta)\to\min_{\theta\in\Theta}J^{\rightarrow}(\theta),
\]
and every accumulation point of minimizers of $J_h^{\rightarrow}$ is a minimizer of $J^{\rightarrow}$.  If the minimizer of $J^{\rightarrow}$ is unique, the entire sequence of minimizers converges.
\end{corollary}

\begin{proof}
The constants in \Cref{thm:euler_conv} are uniform over the compact coefficient family.  Uniform convergence and the standard argmin theorem give the result.
\end{proof}

The adapted topology matters here.  Weak convergence alone does not control filtrations, conditional laws, or the Doob decomposition, whereas adapted Wasserstein convergence is designed to preserve these operations \cite{BackhoffBartlBeiglbockEder2020,BartlBeiglbockPammer2026}.  The synchronous coupling used in \Cref{thm:euler_conv} gives convergence in the stronger adapted metric directly.

\section{Expert-first projection of the arithmetic pool}\label{sec:reverse_projection}

We now minimize the expert-first objective \eqref{eq:reverse_IT_objective} over Markov diffusion coefficients.  This is the diffusion projection associated with the arithmetic likelihood pool.  Unlike the candidate-first problem, the state distributions in the objective are fixed by the experts, so no dynamic programming equation is needed.

Assume in this section that, for each $t>0$, the time-$t$ law of expert $k$ has a continuous density $\rho_k(t,x)$ with respect to a common $\sigma$-finite measure $\nu_t(\dd x)$.  Under Assumption~\ref{ass:regular}, standard uniformly elliptic diffusion theory supplies such densities with $\nu_t$ equal to Lebesgue measure; the separate formulation also accommodates broader weak settings.  We state selectors on the support of the aggregate density; whenever global uniqueness of the Markov coefficient field is asserted, assume additionally that the aggregate density is strictly positive.  Put
\begin{equation}\label{eq:occupancy_weights}
 \rho(t,x):=\sum_{j=1}^K\pi_j\rho_j(t,x),
 \qquad
 \omega_k(t,x):=\frac{\pi_k\rho_k(t,x)}{\rho(t,x)}
\end{equation}
whenever $\rho(t,x)>0$.  These are the posterior expert probabilities conditional on the current state $X_t=x$, as opposed to the full-path posterior weights $w_t^k$ in \eqref{eq:posterior_weights}.

For Markov candidate coefficients $(\beta,A)$, define
\begin{equation}\label{eq:reverse_projection_objective}
 \mathcal J^{\leftarrow}(\beta,A)
 :=\sum_{k=1}^K\pi_k\D_\lambda(P^{m_k}\|P^{\beta,A}).
\end{equation}
Disintegration gives
\begin{align}\label{eq:reverse_disintegration}
 \mathcal J^{\leftarrow}(\beta,A)
 =\int_0^T\!\int \rho(t,x)\Bigg[&
 \frac12\sum_k\omega_k(t,x)
 (b_k-\beta)^\top A^{-1}(b_k-\beta)\\
 &+\frac\lambda2\sum_k\omega_k(t,x)\BW^2(a_k,A)
 \Bigg]\nu_t(\dd x)\dd t.\notag
\end{align}

Define the state-occupancy arithmetic drift and its dispersion matrix by
\begin{equation}\label{eq:reverse_mean_dispersion}
 \bar b^{\omega}(t,x):=\sum_k\omega_k(t,x)b_k(t,x),
 \qquad
 S_b^{\omega}(t,x):=\sum_k\omega_k(t,x)
 (b_k-\bar b^{\omega})(b_k-\bar b^{\omega})^\top(t,x).
\end{equation}

\begin{theorem}[Markov projection of arithmetic consensus]\label{thm:reverse_projection}
Suppose the candidate drift and covariance take values in compact Euclidean-convex sets $\mathcal B(t,x)\subset\R^d$ and $\mathcal A(t,x)\subset\Spp$, and suppose the data in \eqref{eq:occupancy_weights}--\eqref{eq:reverse_mean_dispersion} are continuous on the support of $\rho$.  Then the minimization of \eqref{eq:reverse_projection_objective} decomposes pointwise in $(t,x)$.

If $\bar b^{\omega}(t,x)\in\mathcal B(t,x)$, the unique drift selector is
\begin{equation}\label{eq:reverse_drift_selector}
 \beta^{\leftarrow}(t,x)=\bar b^{\omega}(t,x).
\end{equation}
After substituting this drift, the covariance selector is the unique minimizer over $\mathcal A(t,x)$ of
\begin{equation}\label{eq:reverse_cov_selector}
 A\longmapsto
 \Psi^{\leftarrow}_{t,x}(A)
 :=\frac12\tr\!\left(A^{-1}S_b^{\omega}(t,x)\right)
 +\frac\lambda2\sum_k\omega_k(t,x)\BW^2(a_k(t,x),A).
\end{equation}
No commutativity of the expert covariance matrices is required.  The selectors are unique almost everywhere with respect to the aggregate occupancy measure $\rho(t,x)\nu_t(\dd x)\dd t$.  They are continuous whenever the admissible multifunctions are continuous in the Hausdorff topology.  If $\rho$ is strictly positive and the selected coefficient fields satisfy Assumption~\ref{ass:regular}, they define the unique Markov diffusion minimizer of $\mathcal J^{\leftarrow}$ within that regular coefficient class.
\end{theorem}

\begin{proof}
For fixed $A$,
\begin{align}\label{eq:reverse_ANOVA}
 \sum_k\omega_k(b_k-\beta)^\top A^{-1}(b_k-\beta)
 =\tr(A^{-1}S_b^{\omega})
 +(\beta-\bar b^{\omega})^\top A^{-1}(\beta-\bar b^{\omega}).
\end{align}
Thus \eqref{eq:reverse_drift_selector} is the unique unconstrained minimizer and remains optimal whenever it belongs to $\mathcal B(t,x)$.  The map $A\mapsto\tr(A^{-1}S_b^{\omega})$ is convex on $\Spp$.  For every positive definite $a_k$, the map $A\mapsto\BW^2(a_k,A)$ is strictly convex by the trace-square-root result used in \Cref{prop:selectors}; see \cite[Theorem~7]{BhatiaJainLim2019}.  Hence \eqref{eq:reverse_cov_selector} is strictly convex.  Continuity and compact Euclidean convexity give existence and uniqueness.  The maximum theorem yields continuity of the selector.  Finally, \eqref{eq:reverse_disintegration} shows that pointwise minimization minimizes the integrated objective.
\end{proof}

\begin{remark}[A sufficient condition for Lipschitz projected coefficients]\label{rem:reverse_lipschitz}
The regularity clause in \Cref{thm:reverse_projection} can be verified by a quantitative convexity condition.  Suppose, on a fixed ellipticity box $\underline a I_d\le A\le\overline a I_d$, that the expert densities are locally Lipschitz and the aggregate density is locally bounded away from zero, so that $\omega_k$, $\bar b^\omega$, and $S_b^\omega$ are locally Lipschitz.  If
\[
 S_b^\omega(t,x)\succeq s_0 I_d
\]
on the region of interest, then for every symmetric direction $H$,
\[
 D^2\!\left[\frac12\tr(A^{-1}S_b^\omega)\right][H,H]
 =\tr(A^{-1}HA^{-1}HA^{-1}S_b^\omega)
 \ge \frac{s_0}{\overline a^3}\norm{H}_F^2.
\]
Hence the covariance objective is uniformly strongly convex.  Standard perturbation estimates for strongly convex parametric programs then imply local Lipschitz continuity of the unique covariance selector when the expert coefficients vary locally Lipschitzly and the admissible covariance set is fixed; see, e.g., \cite{BonnansShapiro2000}.  The same argument applies when the full objective in \eqref{eq:reverse_cov_selector} is uniformly strongly convex, even if $S_b^\omega$ is singular. Moving constraint sets require additional regularity of the constrained optimization problem. Together with the required time regularity and growth bounds, the local coefficient estimates give a well-posed projected diffusion; membership in Assumption~\ref{ass:regular} requires the corresponding global bounds.
\end{remark}

At an interior matrix minimizer, differentiation of \eqref{eq:reverse_cov_selector} gives
\begin{equation}\label{eq:reverse_matrix_FOC}
 A^{-1}S_b^{\omega}A^{-1}
 =\lambda\left[
 I_d-\sum_k\omega_k
 a_k^{1/2}(a_k^{1/2}Aa_k^{1/2})^{-1/2}a_k^{1/2}
 \right].
\end{equation}
The left side is the covariance penalty generated by drift disagreement.  It vanishes only when the expert drifts coincide at the state under consideration.

\begin{corollary}[Scalar expert-first projection increases volatility]\label{cor:reverse_scalar}
In one dimension, write $a_k=\sigma_k^2$, $A=\sigma^2$, and define
\[
 \bar\sigma^{\omega}:=\sum_k\omega_k\sigma_k,
 \qquad
 s_b^2:=\sum_k\omega_k(b_k-\bar b^{\omega})^2.
\]
If $\sigma>0$ is unrestricted, the unique expert-first projection volatility solves
\begin{equation}\label{eq:reverse_scalar_equation}
 \lambda(\sigma^{\leftarrow}-\bar\sigma^{\omega})
 =\frac{s_b^2}{(\sigma^{\leftarrow})^3}.
\end{equation}
Hence
\begin{equation}\label{eq:reverse_inflation}
 \sigma^{\leftarrow}=\bar\sigma^{\omega}
 \quad\Longleftrightarrow\quad s_b^2=0,
 \qquad
 \sigma^{\leftarrow}>\bar\sigma^{\omega}
 \quad\text{if }s_b^2>0.
\end{equation}
Under the specified normalization, drift dispersion raises the selected volatility above the posterior mean of expert volatilities.
\end{corollary}

\begin{proof}
After drift minimization, the scalar objective is
\[
 \frac{s_b^2}{2\sigma^2}
 +\frac\lambda2\sum_k\omega_k(\sigma-\sigma_k)^2.
\]
It tends to infinity as $\sigma\downarrow0$ when $s_b^2>0$ and as $\sigma\to\infty$.  Its derivative is
$-s_b^2/\sigma^3+\lambda(\sigma-\bar\sigma^{\omega})$, and its second derivative is positive.  This proves \eqref{eq:reverse_scalar_equation}--\eqref{eq:reverse_inflation}; the case $s_b^2=0$ reduces to the scalar Bures--Wasserstein barycenter.
\end{proof}

\begin{remark}[Dependence on the drift normalization]
\label{rem:normalization_dependence}
The inflation in \Cref{cor:reverse_scalar} uses the second-argument normalization in \eqref{eq:ell_def}. If instead the drift term in $\ellit((b,A),(\bar b,B))$ were normalized by the first covariance $A$, the expert-first pointwise objective would be
\[
 \frac12\sum_k\omega_k(b_k-\beta)^\top a_k^{-1}(b_k-\beta)
 +\frac\lambda2\sum_k\omega_k\BW^2(a_k,A).
\]
For unrestricted coefficients its selectors would be
\begin{equation}\label{eq:alternative_normalization}
 \beta_{\rm alt}=\left(\sum_k\omega_k a_k^{-1}\right)^{-1}
                  \sum_k\omega_k a_k^{-1}b_k,
 \qquad
 A_{\rm alt}=\operatorname{Bar}_{\rm BW}(a_1,\ldots,a_K;\omega).
\end{equation}
Drift dispersion would then have no effect on covariance. Both conventions recover path-space KL on the common-covariance subproblem, so that reduction does not distinguish them. We use the second-argument convention to preserve the Gaussian KL mean term in \Cref{thm:local_gaussian}. The volatility increase is a consequence of this stated loss, not an implication of arithmetic pooling alone. The numeraire argument in \Cref{subsec:numeraire_motivation} motivates retaining both KL directions but does not select the added covariance loss.
\end{remark}

\begin{remark}[Calibration by a scalar volatility tolerance]\label{rem:scalar_tolerance}
In the normalized coordinates used in \Cref{cor:reverse_scalar}, fix a
state $(t,x)$ and put
\[
 u:=\frac{\sigma^{\leftarrow}}{\bar\sigma^{\omega}},
 \qquad
 r:=\frac{s_b^2}{\lambda(\bar\sigma^{\omega})^4}.
\]
Equation \eqref{eq:reverse_scalar_equation} becomes
$u^3(u-1)=r$ with $u\ge1$.  Since $u\mapsto u^3(u-1)$ is strictly
increasing on $[1,\infty)$, for every $\varepsilon>0$ the proportional
volatility tolerance has the exact calibration
\begin{equation}\label{eq:scalar_tolerance_calibration}
 \sigma^{\leftarrow}\le(1+\varepsilon)\bar\sigma^{\omega}
 \quad\Longleftrightarrow\quad
 \lambda\ge
 \frac{s_b^2}
 {(\bar\sigma^{\omega})^4\varepsilon(1+\varepsilon)^3}.
\end{equation}
When $s_b^2=0$, every $\lambda>0$ satisfies the tolerance.  A single
parameter imposes the same tolerance throughout a region by taking
$\lambda$ at least the supremum of the right-hand side there, provided
that supremum is finite.  This pointwise tolerance complements the
integrated covariance budget in \eqref{eq:covariance_budget}; it applies
to the unrestricted scalar selector of \Cref{cor:reverse_scalar}.
\end{remark}

The following family establishes well-posedness and strict drift-dispersion inflation with common expert volatility. The extension to heterogeneous input volatilities is given in \Cref{cor:reverse_OU_heterogeneous}.

\begin{proposition}[A regular Ornstein--Uhlenbeck projection with common expert volatility]
\label{prop:reverse_OU_closure}
Let $d=1$, $\pi_k>0$, and
\[
 b_k(x)=-\kappa x+\mu_k,\qquad a_k=\sigma^2,
 \qquad \kappa,\sigma>0,
\]
with a common deterministic initial state $x_0$.  Put
$\Delta=\max_k\mu_k-\min_k\mu_k$.
The unrestricted expert-first selectors from
\Cref{thm:reverse_projection,cor:reverse_scalar} extend continuously to
$t=0$ and satisfy Assumption~\ref{ass:regular} on every finite horizon.
In particular, their volatility satisfies
\begin{equation}\label{eq:reverse_OU_bounds}
 \sigma\le\sigma^{\leftarrow}(t,x)
 \le\sigma+\frac{\Delta^2}{4\lambda\sigma^3}.
\end{equation}
If $\Delta>0$, the first inequality is strict for every $t>0$ and
$x\in\R$.  Consequently these coefficients define the unique regular
Markov diffusion minimizer of $\mathcal J^{\leftarrow}$ whenever the
admissible coefficient class contains them.
\end{proposition}

\begin{proof}
Write $q_t=e^{-\kappa t}$, $c_t=(1-q_t)/\kappa$, and
$v_t=\sigma^2(1-q_t^2)/(2\kappa)$.  Expert $k$ has marginal law
$N(q_tx_0+c_t\mu_k,v_t)$.  Cancelling the common quadratic term in
the Gaussian densities gives, for $t>0$,
\begin{equation}\label{eq:reverse_OU_weights}
 \omega_k(t,x)=
 \frac{\pi_k\exp\{h_t\mu_k(x-q_tx_0)-\tfrac12c_th_t\mu_k^2\}}
 {\sum_j\pi_j\exp\{h_t\mu_j(x-q_tx_0)-\tfrac12c_th_t\mu_j^2\}},
 \qquad h_t=\frac{c_t}{v_t}=\frac{2}{\sigma^2(1+q_t)}.
\end{equation}
This formula has a smooth extension to $t=0$, with
$\omega_k(0,x_0)=\pi_k$.  Its spatial derivatives are uniformly bounded,
since
\[
 \partial_x\omega_k=h_t\omega_k(\mu_k-\bar\mu^\omega),
 \qquad \bar\mu^\omega=\sum_j\omega_j\mu_j,
 \qquad h_t\le 2/\sigma^2.
\]
On $[0,T]$, differentiating \eqref{eq:reverse_OU_weights} in time also
gives $|\partial_t\omega_k(t,x)|\le C(1+|x|)$.
Hence $\beta^{\leftarrow}=-\kappa x+\bar\mu^\omega$ is globally
Lipschitz in $x$ with linear growth. For
$r:=s_b^2=\sum_k\omega_k(\mu_k-\bar\mu^\omega)^2\le\Delta^2/4$,
\[
 \partial_xr=h_t\sum_k\omega_k(\mu_k-\bar\mu^\omega)^3,
 \qquad |\partial_xr|\le h_t\Delta^3.
\]
Thus $r$ is globally Lipschitz in $x$ and Lipschitz in time with the same linear-growth envelope.
The positive solution $s=s(r)$ of
$\lambda(s-\sigma)s^3=r$ satisfies
\[
 0<s'(r)=\frac{1}{\lambda s^2(4s-3\sigma)}
 \le\frac{1}{\lambda\sigma^3}.
\]
Thus $\sigma^{\leftarrow}=s(r)$ and its square have the asserted
regularity; \eqref{eq:reverse_OU_bounds} gives the uniform ellipticity
and upper bound.  All weights are positive, so $r>0$ when
$\Delta>0$.  The minimizing property follows from the pointwise
decomposition in \Cref{thm:reverse_projection}.
\end{proof}

\begin{remark}[Why the common expert volatility matters]
\label{rem:reverse_OU_restriction}
The common volatility cancels the quadratic state term in every
log-density ratio; \eqref{eq:reverse_OU_weights} shows explicitly why
the remaining slope stays bounded as $t\downarrow0$.
With unequal expert volatilities, the quadratic coefficient can be of
order $t^{-1}$, so Gaussianity alone gives no such uniform estimate.
Nor does Gaussianity ensure a bounded selected covariance at positive times:
two centered Ornstein--Uhlenbeck experts can have identical marginal
variances at a positive time but different mean-reversion rates.
Their posterior weights are then constant in $x$, their drift
dispersion is proportional to $x^2$, and
\eqref{eq:reverse_scalar_equation} gives
$\sigma^{\leftarrow}(t,x)\asymp |x|^{1/2}$.
\end{remark}

\begin{corollary}[Heterogeneous Ornstein--Uhlenbeck volatilities after a positive time]
\label{cor:reverse_OU_heterogeneous}
Let $b_k(x)=-\kappa x+\mu_k$ and $a_k=\sigma_k^2$, where
$\kappa>0$, $\sigma_k>0$, and all experts start from $x_0$ at time zero.
Fix $0<t_0<T$ and retain their original marginal densities
$\rho_k(t,x)$ for $t\in[t_0,T]$. Define
\begin{equation}\label{eq:reverse_truncated_objective}
 \mathcal J^{\leftarrow}_{t_0}(\beta,A)
 :=\sum_k\pi_k\E^{P^{m_k}}\int_{t_0}^T
 \ellit_\lambda\big((b_k,\sigma_k^2),(\beta,A)\big)(t,X_t)\dd t.
\end{equation}
Put $\sigma_-:=\min_k\sigma_k$, $\sigma_+:=\max_k\sigma_k$, and
$\Delta:=\max_k\mu_k-\min_k\mu_k$. The unrestricted selectors
\eqref{eq:reverse_drift_selector} and \eqref{eq:reverse_scalar_equation}
satisfy the coefficient requirements of Assumption~\ref{ass:regular}
on $[t_0,T]$, with constants allowed to depend on $t_0$, and
\begin{equation}\label{eq:reverse_OU_heterogeneous_bounds}
 \sigma_-\le\bar\sigma^\omega\le\sigma^{\leftarrow}
 \le\sigma_++\frac{\Delta^2}{4\lambda\sigma_-^3}.
\end{equation}
If $\Delta>0$, then $\sigma^{\leftarrow}>\bar\sigma^\omega$ everywhere.
For any prescribed candidate initial law at $t_0$ with finite second
moment, these coefficients define the unique regular Markov diffusion
minimizer of \eqref{eq:reverse_truncated_objective}, whenever the
admissible coefficient class contains them.
\end{corollary}

\begin{proof}
With $q_t=e^{-\kappa t}$, $c_t=(1-q_t)/\kappa$, and
$d_t=(1-q_t^2)/(2\kappa)$, expert $k$ has law
$N(q_tx_0+c_t\mu_k,\sigma_k^2d_t)$. Write
$L_{ij}=\log(\pi_i\rho_i/(\pi_j\rho_j))$. For $z=x,t$,
\[
 \partial_z\omega_i=\sum_j\omega_i\omega_j\partial_zL_{ij},
 \qquad \omega_i\omega_j\le e^{-|L_{ij}|}.
\]
If $\sigma_i\ne\sigma_j$, the quadratic coefficient of $L_{ij}$
is bounded away from zero uniformly on $[t_0,T]$; if
$\sigma_i=\sigma_j$ but $\mu_i\ne\mu_j$, its linear coefficient is
$c_t(\mu_i-\mu_j)/(\sigma_i^2d_t)$ and is bounded away from zero.
In either case, exponential decay of $e^{-|L_{ij}|}$ in the tails
dominates the degree-at-most-two polynomials $\partial_zL_{ij}$,
whose coefficients are uniformly bounded on $[t_0,T]$.
For identical expert coefficients, $\partial_zL_{ij}=0$.
Thus all $\partial_x\omega_i$ and $\partial_t\omega_i$ are uniformly
bounded. It follows that
$\beta^{\leftarrow}=-\kappa x+\sum_k\omega_k\mu_k$,
$u:=\bar\sigma^\omega$, and
$r:=\sum_k\omega_k(\mu_k-\sum_j\omega_j\mu_j)^2$
have the required spatial and temporal regularity, with
$0\le r\le\Delta^2/4$ and $\sigma_-\le u\le\sigma_+$.
The solution $s=s(u,r)\ge u$ of $\lambda(s-u)s^3=r$ obeys
\[
 0<\partial_rs=\frac{1}{\lambda s^2(4s-3u)}
 \le\frac{1}{\lambda\sigma_-^3},\qquad
 0<\partial_us=\frac{s}{4s-3u}\le1.
\]
These estimates and \eqref{eq:reverse_OU_heterogeneous_bounds}
give the same regularity for $s$ and $s^2$, as well as bounded
uniform ellipticity. Positivity of the weights gives strict inflation
when $\Delta>0$. Pointwise minimization in
\eqref{eq:reverse_disintegration}, with integration restricted to
$[t_0,T]$, proves optimality and uniqueness.
\end{proof}

The densities in this corollary are those of the original experts
observed after time $t_0$; the experts are not restarted from a common
point at $t_0$. The conclusion concerns the truncated objective
\eqref{eq:reverse_truncated_objective} and does not assert uniform
regularity as $t_0\downarrow0$.

The right panel of \Cref{fig:orientation_volatility_corrections} illustrates the strictly increasing response of the expert-first consensus volatility to drift dispersion.

\begin{proposition}[Mimicking benchmark]\label{prop:mimicking}
Realize the arithmetic pool on an enlarged space by drawing a latent expert index $K$ with probabilities $\pi_k$ and then evolving
\[
 \dd X_t=b_K(t,X_t)\dd t+a_K(t,X_t)^{1/2}\dd W_t.
\]
Under the regularity and nondegeneracy hypotheses of the Markovian mimicking theorem, there exists a Markov diffusion with the same one-time marginals as $P^A$ and coefficients
\begin{equation}\label{eq:mimicking_coefficients}
 \widehat b^{\,\mathrm{mim}}(t,x)
 =\E[b_K(t,X_t)\mid X_t=x]
 =\sum_k\omega_k(t,x)b_k(t,x)=\bar b^\omega(t,x),
\end{equation}
\begin{equation}\label{eq:mimicking_covariance}
 \widehat a^{\,\mathrm{mim}}(t,x)
 =\E[a_K(t,X_t)\mid X_t=x]
 =\sum_k\omega_k(t,x)a_k(t,x).
\end{equation}
Consequently, the mimicking diffusion and the exact arithmetic pool assign the same expectation to every integrable European payoff $f(X_t)$ at each fixed maturity $t$.
\end{proposition}

\begin{proof}
The conditional-expectation formulas follow from Bayes' rule and \eqref{eq:occupancy_weights}.  Existence of a weak Markov diffusion with the same fixed-time marginals follows from \cite{BrunickShreve2013,Gyongy1986}.
\end{proof}

\begin{remark}[Mimicking versus the expert-first projection]\label{rem:mimicking_comparison}
The mimicking benchmark and the expert-first projection agree on drift:
\[
 \widehat b^{\,\mathrm{mim}}=\beta^\leftarrow=\bar b^\omega.
\]
They generally disagree on covariance.  Mimicking uses the posterior arithmetic mean of covariance matrices in \eqref{eq:mimicking_covariance}; the expert-first variational projection uses \eqref{eq:reverse_cov_selector}.  If drift dispersion vanishes, the latter is the Bures--Wasserstein barycenter.  In one dimension,
\[
 a^{\leftarrow}=\left(\sum_k\omega_k\sigma_k\right)^2
 \le \sum_k\omega_k\sigma_k^2=\widehat a^{\,\mathrm{mim}},
\]
with equality only when the expert volatilities coincide.  If drift dispersion is positive, \Cref{cor:reverse_scalar} raises $\sigma^\leftarrow$ above the weighted mean of expert volatilities, and its variance may lie below or above the mimicking variance depending on the size of the dispersion term.

Mimicking is appropriate when the objective is exact preservation of one-time marginals and therefore of European prices.  The expert-first projection answers a different question: it is the unique Markov coefficient field minimizing the expert-first objective, whose drift-dispersion penalty follows from the reference-covariance normalization and \eqref{eq:reverse_ANOVA}.  It need not preserve one-time marginals, but it retains the paper's drift--covariance criterion and reacts to disagreement in drifts, which the mimicking covariance ignores.  The two surrogates can therefore differ for path-dependent pricing, dynamic hedging, and stochastic control even when they agree on the current-state posterior drift.
\end{remark}

\section{Candidate-first barycenters in the diffusion class}\label{sec:barycenter}

We now analyze the candidate-first objective $J^{\rightarrow}$ in \eqref{eq:forward_IT_objective}.  This is the geometric orientation: its common-covariance limiting case is the barycenter in \cite{JaimungalPesenti2026}, while heterogeneous covariance is handled by transport.  Let $m_k=(b_k,a_k)$, $k=1,\ldots,K$, satisfy Assumption \ref{ass:regular}, with $\pi_k>0$ and $\sum_k\pi_k=1$.  A candidate model is controlled through its local characteristics $(\beta_t,A_t)$.  The drift takes values in a compact convex set $\mathcal B(t,x)\subset\R^d$.  The covariance takes values in a compact \emph{Euclidean-convex} set $\mathcal A(t,x)\subset\Spp$ containing the expert covariances.  Canonical choices are an ellipticity box
\[
 \{A\in\mathbb S^d:\underline a I_d\le A\le\overline a I_d\}
\]
or an affine financial slice such as the stochastic-volatility set in \Cref{cor:Heston_slice}.  Euclidean convexity, rather than geodesic convexity for the Bures--Wasserstein metric, is the relevant condition for uniqueness of the Hessian-corrected local selector.

For a controlled weak solution
\[
 \dd X_s=\beta_s\dd s+A_s^{1/2}\dd W_s,
 \qquad (\beta_s,A_s)\in\mathcal B(s,X_s)\times\mathcal A(s,X_s),
\]
define the running barycenter cost
\begin{equation}\label{eq:running_bary}
 L_\lambda(t,x;\beta,A)
 :=\frac12\sum_{k=1}^K\pi_k
 (\beta-b_k)^\top a_k^{-1}(\beta-b_k)(t,x)
 +\frac\lambda2\sum_{k=1}^K\pi_k\BW^2(A,a_k(t,x)).
\end{equation}
The value function is
\begin{equation}\label{eq:value_bary}
 V(t,x):=\inf_{(\beta,A)}
 \E_{t,x}\int_t^T L_\lambda(s,X_s;\beta_s,A_s)\dd s.
\end{equation}

\begin{assumption}[Compact admissible characteristics]\label{ass:control}
The multifunctions $\mathcal B$ and $\mathcal A$ have nonempty compact convex values, measurable graphs, and are continuous in $(t,x)$ in the Hausdorff topology.  There is a constant $C$ such that
\[
 \sup_{\beta\in\mathcal B(t,x)}|\beta|\le C(1+|x|),
 \qquad
 \underline a I_d\le A\le\overline a I_d
 \quad\text{for every }A\in\mathcal A(t,x).
\]
The expert drifts are globally Lipschitz with linear growth, and the expert covariance fields are bounded, globally Lipschitz, and satisfy the same ellipticity bounds. In addition, the admissible sets have a fixed compact control parametrization: there are compact metric spaces $U_b,U_a$ and jointly continuous maps $\beta(t,x,u)$ and $A(t,x,v)$ whose images are $\mathcal B(t,x)$ and $\mathcal A(t,x)$. These maps are globally Lipschitz in $x$, uniformly in time and control, and satisfy \eqref{eq:time_regularity} uniformly in the control parameters. The expert coefficients satisfy the same time bound.
\end{assumption}

\begin{theorem}[Existence and Hamilton--Jacobi--Bellman characterization]\label{thm:HJB}
Under Assumption \ref{ass:control}, the infimum in \eqref{eq:value_bary} is attained by a Markov control.  The value function is the unique continuous viscosity solution with at most quadratic growth of
\begin{equation}\label{eq:HJB}
 \partial_tV+
 \inf_{\substack{\beta\in\mathcal B(t,x)\\A\in\mathcal A(t,x)}}
 \left\{
 \beta^\top\nabla V+\frac12\tr(A D^2V)
 +L_\lambda(t,x;\beta,A)
 \right\}=0,
 \qquad V(T,x)=0.
\end{equation}
If the value function is classical, any measurable minimizing feedback whose controlled equation admits an admissible weak solution is optimal.
\end{theorem}

\begin{proof}
Linear growth of the drift and boundedness of the covariance give, uniformly
over admissible controls,
\[
 \E_{t,x}\sup_{t\le s\le T}|X_s|^q\le C_q(1+|x|^q),\qquad q\ge2.
\]
The running cost is continuous, nonnegative, and bounded above by
$C(1+|x|^2)$. Compactness of the control parametrization gives tightness
for relaxed controls, and these moment bounds permit localization of
the quadratic cost. The epigraph of $L_\lambda$ in the admissible
characteristics $(\beta,A)$ is convex: the drift term is quadratic,
and $\BW^2(\cdot,a_k)$ is convex. Thus averaging relaxed
characteristics preserves the generator and cannot increase the cost.
The relaxed-control existence and Markov-selection result of
\cite{HaussmannLepeltier1990} gives an optimal weak Markov control.
The dynamic programming principle and viscosity characterization follow
from the compact-control theory in \cite{FlemingSoner2006}, with the
same moment bounds used to remove localization. In particular,
$V$ is continuous and $0\le V(t,x)\le C(1+|x|^2)$.

For comparison, set $w(\tau,x)=V(T-\tau,x)$. Its equation is
$\partial_\tau w+\mathcal G(\tau,x,Dw,D^2w)=0$, where
\[
 \mathcal G(\tau,x,p,R)=\sup_{u\in U_b,\,v\in U_a}
 \left\{-\beta^\top p-\tfrac12\tr(AR)
             -L_\lambda(T-\tau,x;\beta,A)\right\},
\]
and the coefficients are evaluated at $(T-\tau,x,u,v)$.
This is the Hamiltonian $G$ in \cite{DaLioLey2006}, with their
$H\equiv0$, $g=\beta$, $f=L_\lambda$, and $c=A^{1/2}/\sqrt2$.
Their conditions (A2)(i)--(iv) hold: the control space is compact
(and embeds as a bounded subset of a Banach space); $g$ is uniformly
Lipschitz in $x$ with linear growth; $f$ is locally uniformly
continuous in $x$ with quadratic growth; and $c$ is uniformly
Lipschitz and bounded, since the square-root map is Lipschitz on
the fixed ellipticity box. The zero initial datum satisfies their
(A3). The comparison theorem explicitly assumes
$|U(t,x)|,|W(t,x)|\le\widehat C(1+|x|^2)$; see
\cite[Theorem~2.1, condition~(16)]{DaLioLey2006}.
This includes quadratic growth, so it applies to the stated solution class and gives uniqueness. Finally, localized It\^o's formula and the moment bounds
give the verification assertion at the classical value function.
\end{proof}

\begin{remark}[Interior classical regularity]\label{rem:classical_regularity}
The drift formula below requires a spatial gradient; the covariance formula requires a spatial Hessian. We state the combined selector at classical points.  Under the standard stronger assumptions of parabolic Bellman regularity theory -- for example, H\"older-smooth data, a fixed uniformly elliptic covariance box, and sufficient regularity of the admissible multifunctions -- the Hamiltonian in \eqref{eq:HJB} is uniformly elliptic and concave in $D^2V$.  Parabolic Evans--Krylov estimates and subsequent Schauder theory then give interior $C^{1+\alpha/2,\,2+\alpha}$ regularity away from the terminal boundary; see \cite{Krylov1987}.  In that regime \Cref{prop:selectors,cor:hessian_selector} hold pointwise on compact interior cylinders rather than merely conditionally on the existence of a classical solution.
\end{remark}

Define
\begin{equation}\label{eq:Gg}
 G(t,x):=\sum_{k=1}^K\pi_k a_k(t,x)^{-1},
 \qquad
 g(t,x):=\sum_{k=1}^K\pi_k a_k(t,x)^{-1}b_k(t,x).
\end{equation}

\begin{proposition}[Optimal local characteristics]\label{prop:selectors}
At a classical point of $V$, if the unconstrained drift minimizer belongs to $\mathcal B(t,x)$, then
\begin{equation}\label{eq:optimal_drift}
 \beta^*(t,x)=G(t,x)^{-1}\big(g(t,x)-\nabla V(t,x)\big).
\end{equation}
The covariance selector is the unique minimizer over the compact convex set $\mathcal A(t,x)$ of
\begin{equation}\label{eq:optimal_cov}
 A\longmapsto
 \frac12\tr\big(A D^2V(t,x)\big)
 +\frac\lambda2\sum_{k=1}^K\pi_k\BW^2(A,a_k(t,x)).
\end{equation}
No commutativity of the covariance matrices is required.
\end{proposition}

\begin{proof}
The drift-dependent part of the Hamiltonian is
\[
 \frac12\beta^\top G\beta-\beta^\top g+\beta^\top\nabla V+\text{constant},
\]
so its first-order condition gives \eqref{eq:optimal_drift}.  For fixed $B\in\Spp$, the map
\[
 A\longmapsto \tr\!\left[(B^{1/2}AB^{1/2})^{1/2}\right]
\]
is strictly concave on $\Spp$.  Indeed, the strict concavity result in \cite[Theorem~7]{BhatiaJainLim2019} states that $C\mapsto\tr(C^{1/2})$ is strictly concave on $\Spp$, and the affine congruence $A\mapsto B^{1/2}AB^{1/2}$ is injective because $B\succ0$.  Therefore $A\mapsto\BW^2(A,B)$ is strictly convex.  The positive weighted sum in \eqref{eq:optimal_cov}, plus its linear Hessian term, remains strictly convex.  Continuity and compactness give existence, while Euclidean convexity gives uniqueness.
\end{proof}

\subsection{The matrix covariance selector}

To give the selector in the original state metric as well, we allow general $M$ in the next theorem. For $A,B\in\Spp$, write
\begin{equation}\label{eq:fidelity}
 \Fid(A,B):=\tr\!\left[(B^{1/2}AB^{1/2})^{1/2}\right].
\end{equation}

\begin{lemma}[Fidelity invariance under paired congruence transformations]\label{lem:fidelity_congruence}
Let $A,B,C\in\Spp$ and set
\[
 X=C^{1/2}AC^{1/2},
 \qquad \widetilde B=C^{-1/2}BC^{-1/2}.
\]
Then
\begin{equation}\label{eq:fidelity_invariance}
 \Fid(A,B)=\Fid(X,\widetilde B).
\end{equation}
\end{lemma}

\begin{proof}
The positive definite matrix $B^{1/2}AB^{1/2}$ is similar to $AB$.  Likewise, $\widetilde B^{1/2}X\widetilde B^{1/2}$ is similar to $X\widetilde B$.  But
\[
 X\widetilde B=C^{1/2}AB C^{-1/2},
\]
so $X\widetilde B$ is similar to $AB$.  The two symmetric positive definite fidelity arguments therefore have the same positive eigenvalues.  Taking positive square roots and summing the eigenvalues proves \eqref{eq:fidelity_invariance}.
\end{proof}

\begin{theorem}[Congruence representation and unrestricted finiteness]
\label{thm:congruence_selector}
Here $\operatorname{Bar}_{\mathrm{BW}}$ denotes the ordinary Bures--Wasserstein barycenter. Fix $B_1,\ldots,B_K\in\Spp$, $M\in\Spp$, positive weights
$\pi_k$ with $\sum_k\pi_k=1$, a symmetric matrix $H$, and $\lambda>0$.
Define
\begin{equation}\label{eq:Psi_H}
 \Psi_{H,M}(A):=\frac12\tr(HA)
 +\frac\lambda2\sum_{k=1}^K\pi_kd_M^2(A,B_k),
 \qquad A\in\Spp,
\end{equation}
and put $D=M+\lambda^{-1}H$.
\begin{enumerate}[label=(\roman*),leftmargin=2em]
\item If $D\succ0$, set
\[
 X=D^{1/2}AD^{1/2},\qquad
 \widetilde B_k=D^{-1/2}MB_kMD^{-1/2}.
\]
Then
\begin{equation}\label{eq:congruence_identity}
 \Psi_{H,M}(A)
 =\frac\lambda2\sum_k\pi_k\BW^2(X,\widetilde B_k)+K_{D,M},
\end{equation}
where
\begin{equation}\label{eq:K_C}
 K_{D,M}:=\frac\lambda2\sum_k\pi_k
 \big\{\tr(MB_k)-\tr(D^{-1}MB_kM)\big\}
\end{equation}
does not depend on $A$.  The unique minimizer is
\begin{equation}\label{eq:Astar_congruence}
 A^*=D^{-1/2}
 \operatorname{Bar}_{\mathrm{BW}}
   (\widetilde B_1,\ldots,\widetilde B_K;\pi)
 D^{-1/2}.
\end{equation}
The transformed barycenter $X^*$ is equivalently the unique positive
definite solution of
\begin{equation}\label{eq:transformed_fixed_point}
 X^*=\sum_k\pi_k
 \big((X^*)^{1/2}\widetilde B_k(X^*)^{1/2}\big)^{1/2}.
\end{equation}
\item If $D\not\succ0$, then
\begin{equation}\label{eq:unbounded_selector}
 \inf_{A\in\Spp}\Psi_{H,M}(A)=-\infty.
\end{equation}
Thus $H+\lambda M\succ0$ is necessary and sufficient for a finite
unrestricted local covariance infimum.
\end{enumerate}
\end{theorem}

\begin{proof}
The Gaussian transport formula implies
\begin{equation}\label{eq:metric_fidelity}
 d_M^2(A,B)=\tr(MA)+\tr(MB)-2\Fid(A,MBM).
\end{equation}
Indeed, the products of the two whitened covariance matrices in
\eqref{eq:metric_bures} and the product $A(MBM)$ are similar; their
positive square-root eigenvalues therefore have the same sum.  Since
the weights sum to one, expansion gives
\begin{equation}\label{eq:metric_selector_expansion}
 \Psi_{H,M}(A)=\frac\lambda2\tr(DA)
 +\frac\lambda2\sum_k\pi_k\tr(MB_k)
 -\lambda\sum_k\pi_k\Fid(A,MB_kM).
\end{equation}
If $D\succ0$, then $\tr(DA)=\tr(X)$ and paired congruence
invariance from \Cref{lem:fidelity_congruence} gives
$\Fid(A,MB_kM)=\Fid(X,\widetilde B_k)$.  Adding and subtracting
$\frac\lambda2\sum_k\pi_k\tr(\widetilde B_k)$ proves
\eqref{eq:congruence_identity}.  The map $A\mapsto X$ is a bijection
of $\Spp$, and all $\widetilde B_k$ are positive definite.  Existence,
uniqueness, and the fixed-point equation for their Bures--Wasserstein
barycenter follow from \cite{AguehCarlier2011,AlvarezEsteban2016}.

If a unit vector $v$ satisfies $v^\top Dv<0$, take
$A_t=I_d+t vv^\top$.  The first term in
\eqref{eq:metric_selector_expansion} tends linearly to $-\infty$,
and the fidelity terms are nonnegative and subtracted.  Hence the
objective tends to $-\infty$.

The remaining case is $D\succeq0$ singular.  Choose a unit vector
$v\in\ker D$ and use the same ray.  The linear term is constant,
whereas, with $R_k=MB_kM\succ0$,
\[
 \Fid(A_t,R_k)
 =\|A_t^{1/2}R_k^{1/2}\|_*
 \ge\|A_t^{1/2}R_k^{1/2}\|_F
 =\sqrt{\tr(A_tR_k)}
 \ge\sqrt{t\,v^\top R_kv}.
\]
Each positively weighted fidelity contribution therefore forces the
objective to $-\infty$.  This proves (ii).
\end{proof}

\begin{remark}[Relation to Gaussian Wasserstein proximal steps]\label{rem:JKO_proximal}
For $M=I_d$ and $K=1$, the Bures--Wasserstein barycenter in transformed coordinates is simply $\widetilde B_1$, and \eqref{eq:Astar_congruence} reduces to
\begin{equation}\label{eq:single_expert_prox}
 A^*=C^{-1}B_1C^{-1},
 \qquad C=I_d+\lambda^{-1}H.
\end{equation}
This is the centered-Gaussian Wasserstein proximal update for the quadratic potential $x\mapsto\frac12x^\top Hx$ with time step $\lambda^{-1}$, using the same Bures--Wasserstein metric as Gaussian variational inference and Wasserstein gradient-flow algorithms; see \cite{DiaoEtAl2023,LambertEtAl2022}.  Tschiderer~\cite{Tschiderer2026} also studies volatility control in Wasserstein gradient-flow representations, for diffusions with a common invariant measure. Here the local multi-expert problem reduces by congruence to a covariance barycenter, with unrestricted finiteness threshold $H+\lambda M\succ0$.
\end{remark}

\begin{proposition}[The constrained covariance selector]
\label{prop:covariance_KKT}
Let $H=H^\top$, $M\succ0$, and $B_k\succ0$, and define
\[
 R_k=MB_kM,\qquad
 T_k(A)=R_k^{1/2}(R_k^{1/2}AR_k^{1/2})^{-1/2}R_k^{1/2}.
\]
The gradient of the objective in \eqref{eq:Psi_H} is
\begin{equation}\label{eq:covariance_gradient}
 \nabla_A\Psi_{H,M}(A)=\frac12H+
 \frac\lambda2\left(M-\sum_k\pi_kT_k(A)\right).
\end{equation}
On a nonempty compact convex set $\mathcal A\subset\Spp$, the unique minimizer $A^*$ is characterized by
\begin{equation}\label{eq:covariance_VI}
 \big\langle\nabla_A\Psi_{H,M}(A^*),A-A^*\big\rangle_F\ge0,
 \qquad A\in\mathcal A,
\end{equation}
where $\langle U,W\rangle_F=\tr(UW)$ for symmetric matrices. On the box
$\underline a I_d\preceq A\preceq\overline a I_d$, with $0<\underline a<\overline a$, the condition is equivalent to the existence of $U_-,U_+\succeq0$ such that
\begin{align}\label{eq:covariance_KKT}
 \nabla_A\Psi_{H,M}(A^*)-U_-+U_+&=0,\\
 \tr[U_-(A^*-\underline a I_d)]
 =\tr[U_+(\overline a I_d-A^*)]&=0.\notag
\end{align}
These characterizations hold whether or not the unrestricted finiteness threshold is satisfied.
\end{proposition}

\begin{proof}
Differentiate the trace-square-root expression in \eqref{eq:metric_selector_expansion} to obtain \eqref{eq:covariance_gradient}. Strict convexity gives existence and uniqueness on the compact set and makes \eqref{eq:covariance_VI} necessary and sufficient. For the box, $\frac12(\underline a+\overline a)I_d$ is strictly feasible. The convex Karush--Kuhn--Tucker conditions therefore give \eqref{eq:covariance_KKT}, including complementary slackness.
\end{proof}

\begin{corollary}[Hessian-corrected covariance barycenter]\label{cor:hessian_selector}
At a classical point $(t,x)$ of the HJB, put $H=D^2V(t,x)$ and $B_k=a_k(t,x)$.  If
\begin{equation}\label{eq:C_condition}
 C(t,x):=I_d+\lambda^{-1}D^2V(t,x)\succ0
\end{equation}
and covariance is locally unrestricted, define
\[
 \widehat a_k(t,x):=C(t,x)^{-1/2}a_k(t,x)C(t,x)^{-1/2}.
\]
Here $\operatorname{Bar}_{\mathrm{BW}}(B_1,\ldots,B_K;\pi)$ denotes the unique positive-definite Bures--Wasserstein barycenter of $B_1,\ldots,B_K$ with weights $\pi$.  Then the unique selector is
\begin{equation}\label{eq:Hessian_cov_formula}
 A^*(t,x)=C(t,x)^{-1/2}
 \operatorname{Bar}_{\mathrm{BW}}\!\left(
 \widehat a_1(t,x),\ldots,\widehat a_K(t,x);\pi
 \right)C(t,x)^{-1/2}.
\end{equation}
If \eqref{eq:C_condition} fails, the unrestricted local Hamiltonian is $-\infty$.  On a compact Euclidean-convex covariance set, the selector is the unique solution of the variational inequality in \Cref{prop:covariance_KKT}, including boundary cases.
\end{corollary}

In the original coordinates, if $V_x(t,x)=V_y(t,M^{1/2}x)$ and $H_x=D_x^2V_x$, then
\[
 D_y^2V_y=M^{-1/2}H_xM^{-1/2},\qquad
 I+\lambda^{-1}D_y^2V_y\succ0
 \quad\Longleftrightarrow\quad H_x+\lambda M\succ0.
\]
The selector in those coordinates is given directly by \Cref{thm:congruence_selector}, with $H=H_x$ and the original expert covariances. It equals $M^{-1/2}A_y^*M^{-1/2}$ by uniqueness and linear equivariance.

The threshold compares continuation-value curvature with the covariance-loss metric. Sufficiently negative curvature makes the unrestricted local objective unbounded below; even a null direction at equality is enough because of the negative square-root term in \eqref{eq:metric_selector_expansion}. Compact covariance restrictions keep the Hamiltonian finite. Failure of the threshold forces a boundary solution on a full-dimensional compact covariance set, but an affine financial restriction may admit a minimizer in its relative interior.

\begin{remark}[The role of $\lambda$]\label{rem:lambda_limits}
For a fixed matrix $H$ and a fixed compact convex covariance set, as $\lambda\to\infty$ the selector converges to the constrained Bures--Wasserstein barycenter.  As $\lambda\downarrow0$, accumulation points minimize $A\mapsto\tr(HA)$ and need not be unique when $H$ has flat directions.  In the unrestricted problem,
\[
 I_d+\lambda^{-1}H\succ0
 \quad\Longleftrightarrow\quad
 H\succ-\lambda I_d.
\]
Hence, if $H\succ0$, the congruence formula remains valid for every $\lambda>0$ and the minimizer approaches the boundary $A=0$ as $\lambda\downarrow0$.  If $\lambda_{\min}(H)<0$, unrestricted finiteness fails once $0<\lambda\le-\lambda_{\min}(H)$.  When $H\succeq0$ is singular, the condition holds for every $\lambda>0$, but the flat directions are selected by the Bures--Wasserstein term and the limit need not collapse to zero.  These are local comparisons with $H$ held fixed. In the full control problem the value Hessian also depends on $\lambda$; the integrated comparison valid for that problem is \Cref{prop:budget_tradeoff}.
\end{remark}

\subsection{Two limiting cases}

If the expert coefficients and admissible characteristic sets depend only on time, the minimum running cost is deterministic and can be attained pointwise in time. The value function is then independent of $x$, so $\nabla V=D^2V=0$.

\begin{corollary}[Explicit deterministic barycenter]\label{cor:det_bary}
Suppose $b_k=b_k(t)$, $a_k=a_k(t)$, and the admissible sets depend only on time. Assume that the drift in \eqref{eq:det_drift} belongs to $\mathcal B(t)$ and that $\mathcal A(t)$ contains the unconstrained Bures--Wasserstein barycenter.  Then
\begin{equation}\label{eq:det_drift}
 b^*(t)=\left(\sum_k\pi_k a_k(t)^{-1}\right)^{-1}
 \left(\sum_k\pi_k a_k(t)^{-1}b_k(t)\right),
\end{equation}
and $a^*(t)$ is the ordinary Bures--Wasserstein barycenter of $a_1(t),\ldots,a_K(t)$.  In one dimension, writing $a_k=\sigma_k^2$,
\begin{equation}\label{eq:scalar_vol_bary}
 \sigma^*(t)=\sum_{k=1}^K\pi_k\sigma_k(t),
 \qquad
 a^*(t)=\left(\sum_{k=1}^K\pi_k\sqrt{a_k(t)}\right)^2.
\end{equation}
If the covariance set is active, the same statement holds with the unique constrained Bures--Wasserstein minimizer.
\end{corollary}

\begin{corollary}[Common-covariance candidate-first case]\label{cor:JP}
Suppose $a_k(t,x)=a_0(t,x)$ for every $k$ and the drift in \eqref{eq:JP_drift} is admissible. Then the covariance set may be restricted to the singleton $\{a_0(t,x)\}$, the Bures--Wasserstein term vanishes, and \eqref{eq:HJB} reduces to
\begin{align}\label{eq:JP_HJB}
 \partial_tV+\bar b^\top\nabla V+\frac12\tr(a_0D^2V)
 -\frac12\nabla V^\top a_0\nabla V
 +\frac12\sum_k\pi_k(b_k-\bar b)^\top a_0^{-1}(b_k-\bar b)=0,
\end{align}
with $V(T,x)=0$, $\bar b=\sum_k\pi_kb_k$, and
\begin{equation}\label{eq:JP_drift}
 b^*(t,x)=\bar b(t,x)-a_0(t,x)\nabla V(t,x).
\end{equation}
These are exactly the pure Kullback--Leibler barycenter equations in \cite{JaimungalPesenti2026}.
\end{corollary}

The common-covariance candidate-first case should be contrasted with the expert-first arithmetic construction in \Cref{thm:dynamic_arithmetic_pool}.  Candidate-first Kullback--Leibler produces the gradient-corrected Markov drift \eqref{eq:JP_drift} and a normalized geometric density.  Expert-first log-wealth aggregation produces the arithmetic mixture and the path-posterior drift \eqref{eq:dynamic_arithmetic_drift}; its Markov projection is \eqref{eq:reverse_drift_selector}.  These are different economic aggregation rules even before covariance transport is introduced.

\Cref{thm:AW_reduction,cor:JP} are two limiting cases of the candidate-first construction: adapted Bures--Wasserstein transport cost when only covariance moves, and geometric Kullback--Leibler aggregation under a common covariance when only drift moves.  The congruence formula \eqref{eq:Hessian_cov_formula} describes its interior.  The expert-first arithmetic projection is characterized separately in \Cref{sec:reverse_projection}.

\subsection{A scalar linear--quadratic example with heterogeneous Ornstein--Uhlenbeck experts}\label{sec:LQ_example}

Consider scalar expert models
\begin{equation}\label{eq:LQ_experts}
 \dd X_t=-\kappa_kX_t\dd t+\sigma_k\dd W_t^k,
 \qquad \kappa_k>0,\quad\sigma_k>0,
\end{equation}
and candidate controls
\[
 \dd X_t=\beta_t\dd t+\sigma_t\dd W_t,
 \qquad \sigma_t>0.
\]
The admissible class consists of progressively measurable controls for which the state equation is well defined and the running cost is finite.  Define
\begin{equation}\label{eq:LQ_constants}
 G:=\sum_k\frac{\pi_k}{\sigma_k^2},
 \qquad
 q:=\sum_k\frac{\pi_k\kappa_k}{\sigma_k^2},
 \qquad
 c:=\sum_k\frac{\pi_k\kappa_k^2}{\sigma_k^2},
 \qquad
 \bar\sigma:=\sum_k\pi_k\sigma_k.
\end{equation}

\begin{theorem}[Explicit candidate-first volatility aggregation]\label{thm:LQ_solution}
The value function has the quadratic form
\begin{equation}\label{eq:LQ_value}
 V(t,x)=\frac12p_t x^2+r_t.
\end{equation}
The coefficient $p$ is the nonnegative solution of
\begin{equation}\label{eq:LQ_riccati}
 \dot p_t+c-\frac{(q+p_t)^2}{G}=0,
 \qquad p_T=0.
\end{equation}
Writing $\chi:=\sqrt{cG}$, Cauchy--Schwarz gives $q\le\chi$.  If $q<\chi$, then
\begin{equation}\label{eq:p_explicit}
 p_t=\chi\tanh\!\left(
 \sqrt{\frac cG}(T-t)+\operatorname{arctanh}\frac q\chi
 \right)-q;
\end{equation}
if $q=\chi$, then $p\equiv0$.  The optimal local characteristics are
\begin{equation}\label{eq:LQ_controls}
 \beta^*(t,x)=-\frac{q+p_t}{G}x,
 \qquad
 \boxed{\displaystyle \sigma_t^*=\frac{\lambda\bar\sigma}{\lambda+p_t}}.
\end{equation}
Finally, $r_T=0$ and
\begin{equation}\label{eq:r_ode}
 \dot r_t+\frac\lambda2\sum_k\pi_k\sigma_k^2
 -\frac{\lambda^2\bar\sigma^2}{2(\lambda+p_t)}=0.
\end{equation}
Moreover $0\le p_t\le p_0$.  Hence, with
\begin{equation}\label{eq:LQ_compact_bound}
 \underline\sigma_*:=\frac{\lambda\bar\sigma}{\lambda+p_0}>0,
\end{equation}
we have
\[
 \underline\sigma_*\le\sigma_t^*\le\bar\sigma,
 \qquad 0\le t\le T.
\]
If the experts do not all have the same mean-reversion coefficient on the support of the weights, then $p_t>0$ for every $t<T$ and
\begin{equation}\label{eq:vol_shrink}
 0<\sigma_t^*<\bar\sigma=\sum_k\pi_k\sigma_k,
 \qquad t<T,
\end{equation}
while $\sigma_T^*=\bar\sigma$.  If all $\kappa_k$ are equal, then $p\equiv0$ and $\sigma_t^*=\bar\sigma$ for all $t$, even when the $\sigma_k$ differ.  Thus the Hessian channel transmits \emph{drift} heterogeneity into the volatility selector; volatility disagreement by itself does not activate this correction in the present LQ model.
\end{theorem}

\begin{proof}
Insert the ansatz \eqref{eq:LQ_value} into the HJB.  The drift-dependent terms are
\[
 \frac12G\beta^2+(q+p_t)x\beta+\frac12cx^2,
\]
whose unique minimum is attained at the first control in \eqref{eq:LQ_controls} and contributes $\frac12[c-(q+p_t)^2/G]x^2$.  The volatility-dependent terms are
\[
 \frac12p_t\sigma^2+\frac\lambda2\sum_k\pi_k(\sigma-\sigma_k)^2
 =\frac12(\lambda+p_t)\sigma^2-\lambda\bar\sigma\sigma
 +\frac\lambda2\sum_k\pi_k\sigma_k^2.
\]
Once $p_t\ge0$ is established below, this expression is strictly convex and its minimizer is the second control in \eqref{eq:LQ_controls}.  Matching the quadratic and constant coefficients gives \eqref{eq:LQ_riccati} and \eqref{eq:r_ode}.

To solve the Riccati equation, set $s=T-t$ and $u_s=q+p_{T-s}$.  Then
\[
 u'_s=c-\frac{u_s^2}{G},
 \qquad u_0=q.
\]
Cauchy--Schwarz gives $q^2\le Gc=\chi^2$.  Therefore $u'_0\ge0$, and the scalar flow leaves the interval $[q,\chi]$ invariant.  It follows that $u_s$ is nondecreasing in backward time $s$ and hence
\[
 p_{T-s}=u_s-q\ge0.
\]
If $q<\chi$, then $u'_0>0$ and $p_t>0$ for every $t<T$; if $q=\chi$, the equilibrium solution is $u\equiv\chi$ and $p\equiv0$.  The explicit formula \eqref{eq:p_explicit} also shows that $p_t$ decreases monotonically in calendar time, so $0\le p_t\le p_0$.  Equality $q^2=Gc$ in Cauchy--Schwarz holds if and only if all $\kappa_k$ are equal on the support of the positive weights.  The volatility levels $\sigma_k$ affect the weights $1/\sigma_k^2$ but not this equality condition.  Equations \eqref{eq:LQ_compact_bound} and \eqref{eq:vol_shrink} now follow from \eqref{eq:LQ_controls}.

The preceding bounds also close the admissibility issue.  Put
\[
 K_*:=\max\left\{\max_k\kappa_k,\frac{q+p_0}{G}\right\},
 \qquad
 \underline\sigma_c:=\min\left\{\underline\sigma_*,\min_k\sigma_k\right\},
 \qquad
 \overline\sigma_c:=\max_k\sigma_k,
\]
\[
 \mathcal B_*(x):=[-K_*|x|,K_*|x|],
 \qquad
 \mathcal A_*:=[\underline\sigma_c^2,\overline\sigma_c^2].
\]
These are compact convex pointwise control sets with linear growth, contain the expert characteristics, and contain the Hamiltonian minimizers for every $(t,x)$.  Restricting to them therefore does not change the pointwise Hamiltonian, and they satisfy Assumption \ref{ass:control}.

For completeness, the unrestricted problem can be verified directly.  Write $J_{t,x}(\beta,\sigma)$ for its expected running cost.  Applying It\^o's formula to \eqref{eq:LQ_value}, localizing, and completing squares yields
\begin{equation}\label{eq:LQ_verification}
\begin{aligned}
 J_{t,x}(\beta,\sigma)-V(t,x)
 =\E_{t,x}\int_t^T
 \Bigg[&\frac G2\left(\beta_s+\frac{q+p_s}{G}X_s\right)^2\\
 &+\frac{\lambda+p_s}{2}
 \left(\sigma_s-\frac{\lambda\bar\sigma}{\lambda+p_s}\right)^2\Bigg]
 \dd s\ge0.
\end{aligned}
\end{equation}
The feedback control \eqref{eq:LQ_controls} makes both squares vanish.  Its drift is linear with bounded deterministic coefficient and its volatility lies in the compact interval above, so its state has finite moments.

It remains to justify removal of localization for an arbitrary unrestricted admissible control with finite cost.  Finite volatility cost implies
\[
 \E\int_t^T\sigma_s^2\dd s<\infty,
\]
because $\sum_k\pi_k(\sigma_s-\sigma_k)^2$ controls $\sigma_s^2$ up to a deterministic constant.  If two mean-reversion coefficients $\kappa_i\ne\kappa_j$ have positive weights, finite drift cost and
\[
 (\kappa_i-\kappa_j)X_s
 =(\beta_s+\kappa_iX_s)-(\beta_s+\kappa_jX_s)
\]
give $\E\int_t^T X_s^2\dd s<\infty$; then
$\beta_s^2\le2(\beta_s+\kappa_iX_s)^2+2\kappa_i^2X_s^2$ gives
$\E\int_t^T\beta_s^2\dd s<\infty$.  The state equation and the Burkholder--Davis--Gundy inequality therefore yield
\[
 \E\sup_{t\le s\le T}|X_s|^2<\infty.
\]
Since $p$ and $r$ are bounded, the localized terminal terms in It\^o's formula are uniformly integrable and converge to $V(T,X_T)=0$.  If all $\kappa_k$ coincide, then $p\equiv0$ and the value function is deterministic, so no quadratic terminal term is present.  Thus \eqref{eq:LQ_verification} is valid over the full finite-cost admissible class.  This proves optimality both for the compact formulation of \Cref{thm:HJB} and for the original unrestricted problem.
\end{proof}

\begin{figure}[htbp]
 \centering
 \includegraphics[width=0.485\textwidth]{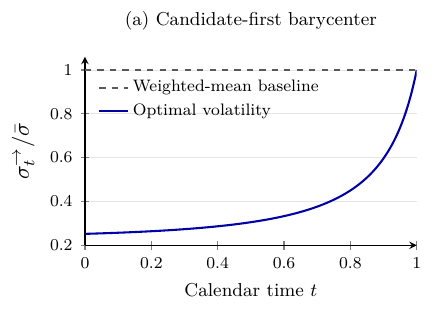}\hfill
 \includegraphics[width=0.485\textwidth]{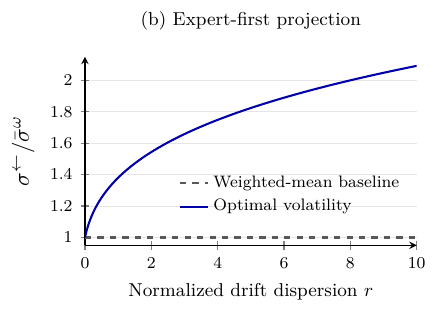}
 \caption{Effect of drift disagreement on the covariance selectors.  Left: the candidate-first ratio $\sigma_t^{\rightarrow}/\bar\sigma$ from \Cref{thm:LQ_solution}; mean-reversion disagreement makes it smaller than the pointwise weighted mean of expert volatilities before maturity.  Parameters are $T=1$, $\lambda=1$, $\pi=(0.3,0.4,0.3)$, $\kappa=(0.5,1.5,3)$, and $\sigma=(0.2,0.35,0.5)$.  Right: the expert-first scalar projection from \Cref{cor:reverse_scalar}; the normalized volatility $\sigma^{\leftarrow}/\bar\sigma^{\omega}$ increases above one with normalized drift dispersion $s_b^2/[\lambda(\bar\sigma^{\omega})^4]$.}
 \label{fig:orientation_volatility_corrections}
\end{figure}

In this linear--quadratic model, $p_t\ge0$, so the unrestricted finiteness threshold holds for every $\lambda>0$. The example illustrates the Hessian correction but does not exhibit threshold failure.

\begin{remark}[Matrix Ornstein--Uhlenbeck companion]
For $d$-dimensional linear experts with constant covariance matrices $A_k$ and a quadratic value function $V(t,x)=\frac12x^\top P_tx+r_t$, the covariance selector is, whenever $C_t=I+\lambda^{-1}P_t\succ0$,
\[
 A_t^*=C_t^{-1/2}
 \operatorname{Bar}_{\mathrm{BW}}\!\left(
 C_t^{-1/2}A_1C_t^{-1/2},\ldots,C_t^{-1/2}A_KC_t^{-1/2};\pi
 \right)C_t^{-1/2}.
\]
Thus the value Hessian changes covariance coordinates before aggregation.  The scalar formula \eqref{eq:LQ_controls} is its one-dimensional specialization.
\end{remark}

\subsection{Noncommuting covariance slices in a stochastic-volatility model}

At a fixed variance state $v>0$, the covariance matrix in \cite{Heston1993} has the form
\begin{equation}\label{eq:Heston_matrix}
 A(v;\rho,\xi)=v
 \begin{pmatrix}
 1&\rho\xi\\
 \rho\xi&\xi^2
 \end{pmatrix}.
\end{equation}
Different $(\rho,\xi)$ generally produce noncommuting matrices.

\begin{corollary}[Unique selector on a bounded stochastic-volatility slice]\label{cor:Heston_slice}
Fix $0<\underline a<\overline a$ and define
\begin{equation}\label{eq:Heston_slice}
 \mathcal H_v:=\left\{A\in\mathbb S_{++}^2:
 A_{11}=v,\quad \underline a I_2\le A\le\overline a I_2\right\}.
\end{equation}
Whenever $\mathcal H_v$ is nonempty, it is compact and Euclidean-convex.  For every symmetric Hessian $H$ and every collection of positive definite expert matrices $B_k$, the objective $\Psi_{H,I}$ in \eqref{eq:Psi_H} has a unique minimizer on $\mathcal H_v$.  That minimizer determines unique effective stochastic-volatility parameters
\begin{equation}\label{eq:Heston_recover}
 \xi^*=\sqrt{A_{22}^*/v},
 \qquad
 \rho^*=\frac{A_{12}^*}{\sqrt{vA_{22}^*}}.
\end{equation}
No simultaneous diagonalization or commutativity assumption is needed.
\end{corollary}

\begin{proof}
The set is the intersection of an affine hyperplane with a compact Loewner interval and is therefore compact and convex.  Strict convexity from \Cref{prop:selectors} gives uniqueness.  Positive definiteness implies $|A_{12}^*|<\sqrt{vA_{22}^*}$, so \eqref{eq:Heston_recover} is well defined and $|\rho^*|<1$.
\end{proof}

\section{Exact transition kernels and an additional coupling term}\label{sec:OU}

The Euler identity \eqref{eq:euler_exact_identity} is exactly additive by construction: it freezes the local Gaussian coefficients over each step.  The finite part of the transition-kernel entropy built from the \emph{exact} transition kernels need not be additive.  The following scalar Ornstein--Uhlenbeck calculation identifies the resulting gap explicitly.  The HJB and linear--quadratic examples above belong to the Euler/local-characteristic problem; the present section records the additional information that enters when Euler kernels are replaced by exact transitions.

Let $Q$ be the law of
\begin{equation}\label{eq:OU_Q}
 \dd X_t=(\mu-\kappa X_t)\dd t+\sigma\dd W_t,
\end{equation}
and $P$ the law of
\begin{equation}\label{eq:OU_P}
 \dd X_t=(\mu_0-\kappa_0X_t)\dd t+\sigma_0\dd W_t,
\end{equation}
with a common deterministic initial value.  Let $Q^N,P^N$ be their laws on the uniform grid $\{jT/N:j=0,\ldots,N\}$ and set $r=\sigma^2/\sigma_0^2$,
\[
 \mathcal K(r)=\frac12\{r-1-\log r\}.
\]

\begin{proposition}[Two-scale entropy of Ornstein--Uhlenbeck skeletons]\label{prop:OU_expansion}
As $N\to\infty$,
\begin{align}\label{eq:OU_entropy}
 \KL(Q^N\|P^N)
 &=N\mathcal K(r)
 +\frac12\E^Q\int_0^T
 \left\{
 \tfrac{\big[(\mu-\mu_0)-(\kappa-\kappa_0)X_t\big]^2}{\sigma_0^2}
 +(r-1)(\kappa_0-\kappa)
 \right\}\dd t
 +O(N^{-1}).
\end{align}
The squared-drift integrand is the information term appearing in $\D_\lambda$. The remaining integrand is an exact drift--volatility coupling correction.  It vanishes if either $\sigma=\sigma_0$ or $\kappa=\kappa_0$, but it may have either sign in general.
\end{proposition}

\begin{proof}
For a step $h=T/N$, the conditional transition laws are Gaussian with
\[
 m_h(x)=e^{-\kappa h}x+\frac\mu\kappa(1-e^{-\kappa h}),
 \qquad
 v_h=\frac{\sigma^2}{2\kappa}(1-e^{-2\kappa h}),
\]
and analogous quantities $m_{0,h},v_{0,h}$ for $P$.  The chain rule for relative entropy gives
\[
 \KL(Q^N\|P^N)=\sum_{j=0}^{N-1}
 \E^Q\KL\big(N(m_h(X_{jh}),v_h)\|N(m_{0,h}(X_{jh}),v_{0,h})\big).
\]
Uniformly on sets with polynomial growth,
\[
 \frac{v_h}{v_{0,h}}
 =r\{1+(\kappa_0-\kappa)h+O(h^2)\},
\]
and
\[
 m_h(x)-m_{0,h}(x)
 =\big[(\mu-\mu_0)-(\kappa-\kappa_0)x\big]h+O((1+|x|)h^2).
\]
Substitution into the Gaussian Kullback--Leibler formula yields, per step,
\begin{align*}
 \mathcal K(r)+\frac h2\left\{
 \frac{\big[(\mu-\mu_0)-(\kappa-\kappa_0)x\big]^2}{\sigma_0^2}
 +(r-1)(\kappa_0-\kappa)
 \right\}+O((1+x^2)h^2).
\end{align*}
Summation, the Ornstein--Uhlenbeck moment bounds, and Riemann convergence prove \eqref{eq:OU_entropy}.
\end{proof}

\begin{remark}[The cross term is informative]\label{rem:OU_gap}
The centered entropy
$\KL(Q^N\|P^N)-N\mathcal K(r)$ is not nonnegative and therefore is not itself a divergence.  The signed coupling term comes from the order-$h^2$ correction to the exact conditional variance, which depends on mean reversion.  Accordingly, $\D_\lambda$ is claimed to be exact for the renormalized \emph{Euler-kernel} functional in \eqref{eq:euler_exact_identity}, not for the finite part of every exact diffusion transition-kernel entropy.  It is the nonnegative local-characteristic baseline; exact transition kernels may add a signed drift--volatility correction.  In non-Gaussian state-dependent models, analogous terms should be expected.  Identifying them by heat-kernel asymptotics is a separate stochastic-analysis problem.
\end{remark}

\section{Financial interpretation and martingale constraints}\label{sec:finance}

For portfolio or policy decisions, expert-first projection evaluates disagreement under the experts' state distributions, whereas candidate-first aggregation accounts for how its own coefficients change future exposure to disagreement. The likelihood-payoff interpretation requires attainable claims to yield an implemented portfolio. The covariance metric and budget specify additional approximation preferences; neither follows from the payoff identity.

A pricing model requires a further restriction: discounted traded assets must remain martingales. For example, if
\[
 \frac{\dd S_t}{S_t}=r_t\dd t+\sigma_t\dd W_t^Q,
\]
then the drift of the discounted log price is $-\sigma_t^2/2$. Independent choices of its drift and covariance need not satisfy this relation. Consider deterministic scalar experts with volatilities $\sigma_k>0$, and use $M=1$ in discounted log-price coordinates. The candidate-first pointwise objective under the martingale restriction is

\begin{equation}\label{eq:martingale_constrained}
 \sum_k\pi_k\left\{
 \frac{(\sigma^2-\sigma_k^2)^2}{8\sigma_k^2}
 +\frac\lambda2(\sigma-\sigma_k)^2
 \right\}.
\end{equation}
Its positive stationary point solves
\begin{equation}\label{eq:cubic_sigma}
 \left(\sum_k\frac{\pi_k}{\sigma_k^2}\right)\sigma^3
 +(2\lambda-1)\sigma
 -2\lambda\sum_k\pi_k\sigma_k=0.
\end{equation}
The positive solution of \eqref{eq:cubic_sigma} is unique for every $\lambda>0$; see \Cref{app:cubic}.  The derivative of \eqref{eq:martingale_constrained} is one half of the left-hand side of \eqref{eq:cubic_sigma}.  The sign analysis in \Cref{app:cubic} therefore shows that the objective decreases before the unique root and increases after it, so this root is the unique global minimizer.  As $\lambda\to\infty$, it approaches the weighted arithmetic mean of expert volatilities $\sum_k\pi_k\sigma_k$.  At the opposite extreme,
\begin{equation}\label{eq:cubic_small_lambda_limit}
 \sigma^2\longrightarrow
 \left(\sum_k\frac{\pi_k}{\sigma_k^2}\right)^{-1}
 \qquad\text{as }\lambda\downarrow0,
\end{equation}
the weighted harmonic mean of the expert variances, equivalently the reciprocal of the weighted mean precision.  Thus pure log-drift matching under the martingale restriction selects harmonic-mean variance, while overwhelming weight on covariance disagreement selects arithmetic-mean volatility.  For intermediate $\lambda$, the martingale drift restriction generates a continuous economically meaningful interpolation.

\section{A localized square-root volatility-factor example}\label{sec:cir}

Consider the square-root variance factor in \cite{CIR1985} under expert $k$,
\begin{equation}\label{eq:CIR_k}
 \dd V_t=\kappa_k(\theta_k-V_t)\dd t+\xi_k\sqrt{V_t}\dd B_t^k,
 \qquad V_t>0.
\end{equation}
A candidate model has parameters $(\kappa,\theta,\xi)$.  Different values of $\xi$ prescribe different quadratic variation,
\[
 [V]_t=\int_0^t\xi^2V_s\dd s,
\]
so the corresponding path laws are mutually singular whenever the pathwise identities differ with probability one.  A pure path-space Kullback--Leibler barycenter is therefore unavailable.

The square-root diffusion is not uniformly elliptic at zero.  We now make precise the localization used to extend the local-characteristic divergence to this setting.  Fix positive bounds
\[
 0<\underline\kappa<\overline\kappa,\qquad
 0<\underline\theta<\overline\theta,\qquad
 0<\underline\xi<\overline\xi,
\]
and $\delta>0$, and define the compact parameter set
\begin{equation}\label{eq:CIR_parameter_set}
 \Theta_\delta:=\left\{(\kappa,\theta,\xi):
 \begin{array}{c}
 \underline\kappa\le\kappa\le\overline\kappa,
 \quad \underline\theta\le\theta\le\overline\theta,
 \quad \underline\xi\le\xi\le\overline\xi,\\[1mm]
 2\kappa\theta-\xi^2\ge\delta
 \end{array}\right\}.
\end{equation}
Thus the candidate set is separated from the Feller boundary, rather than merely contained in its interior.  Assume that all expert parameters lie in $\Theta_\delta$ and that the candidate starts in its stationary distribution.

At state $v>0$, write
\begin{equation}\label{eq:CIR_local_integrand}
 \ell_k(v;\kappa,\theta,\xi)
 :=\frac{\big(\kappa(\theta-v)-\kappa_k(\theta_k-v)\big)^2}
 {2\xi_k^2v}
 +\frac\lambda2(\xi-\xi_k)^2v.
\end{equation}
For $n\ge1$, let
\begin{equation}\label{eq:CIR_tau}
 \tau_n:=\inf\{t\ge0:V_t\notin[n^{-1},n]\}
\end{equation}
and define the stopped objective
\begin{equation}\label{eq:CIR_truncated}
 \D_{\lambda}^{(n)}(\vartheta\|\vartheta_k)
 :=\E^{\vartheta}\int_0^{T\wedge\tau_n}
 \ell_k(V_t;\vartheta)\dd t,
 \qquad \vartheta=(\kappa,\theta,\xi).
\end{equation}
On $[n^{-1},n]$, the Cox--Ingersoll--Ross coefficients can be extended outside the interval to globally Lipschitz, uniformly elliptic coefficients.  By pathwise uniqueness, the extended and original processes agree up to $\tau_n$, so \eqref{eq:CIR_truncated} is covered by the uniformly elliptic local-characteristic construction.

\begin{proposition}[Localization away from the Feller boundary]\label{prop:CIR_localization}
For every $\vartheta,\vartheta_k\in\Theta_\delta$,
\begin{equation}\label{eq:CIR_localization_limit}
 \D_{\lambda}^{(n)}(\vartheta\|\vartheta_k)
 \uparrow
 \E^{\vartheta}\int_0^T\ell_k(V_t;\vartheta)\dd t
 <\infty.
\end{equation}
The limiting stationary objective is continuous on $\Theta_\delta\times\Theta_\delta$.  Consequently every weighted finite-expert barycenter over $\Theta_\delta$ exists.
\end{proposition}

\begin{proof}
For parameters in $\Theta_\delta$, expansion of \eqref{eq:CIR_local_integrand} gives the uniform bound
\begin{equation}\label{eq:CIR_domination}
 0\le \ell_k(v;\vartheta)\le C\left(1+v+v^{-1}\right),
 \qquad v>0,
\end{equation}
with $C$ independent of $\vartheta$ and $\vartheta_k$.  The stationary law is Gamma with shape
$\alpha=2\kappa\theta/\xi^2$.  From \eqref{eq:CIR_parameter_set},
\[
 \inf_{\Theta_\delta}\alpha
 \ge 1+\delta/\overline\xi^2>1.
\]
Thus $\E[V]+\E[V^{-1}]$ is finite, uniformly over $\Theta_\delta$, and \eqref{eq:CIR_domination} makes the full local-cost integral integrable.  The strict Feller condition implies that $V$ is positive and nonexplosive.  Hence, on every finite horizon, $\tau_n\uparrow\infty$ almost surely.  Since the integrand is nonnegative, monotone convergence proves \eqref{eq:CIR_localization_limit}.

The stationary Gamma density and the integrand depend continuously on the parameters.  The explicit stationary moments below, together with the denominator bound $2\kappa\theta-\xi^2\ge\delta$, give continuity of the limiting objective directly.  Compactness of $\Theta_\delta$ then gives existence of a minimizer.
\end{proof}

The stationary moments are
\begin{equation}\label{eq:CIR_stationary_moments}
 \E[V]=\theta,
 \qquad
 \E[V^{-1}]=\frac{2\kappa}{2\kappa\theta-\xi^2}.
\end{equation}
They yield a closed-form objective.

\begin{proposition}[Stationary Cox--Ingersoll--Ross divergence]\label{prop:CIR_formula}
Let
\[
 A_k:=\kappa\theta-\kappa_k\theta_k,
 \qquad B_k:=\kappa-\kappa_k.
\]
For $\vartheta,\vartheta_k\in\Theta_\delta$, the divergence rate from the candidate Cox--Ingersoll--Ross model to expert $k$ is
\begin{align}\label{eq:CIR_div}
 \frac1T\D_\lambda(P^{\kappa,\theta,\xi}\|P^{\kappa_k,\theta_k,\xi_k})
 &=\frac{1}{2\xi_k^2}
 \left\{
 A_k^2\frac{2\kappa}{2\kappa\theta-\xi^2}
 -2A_kB_k+B_k^2\theta
 \right\}
 +\frac\lambda2(\xi-\xi_k)^2\theta.
\end{align}
The weighted Cox--Ingersoll--Ross barycenter over $\Theta_\delta$ is therefore the minimizer of the explicit continuous objective obtained by summing \eqref{eq:CIR_div} with weights $\pi_k$.
\end{proposition}

\begin{proof}
At state $v$,
\[
 b(v)-b_k(v)=A_k-B_kv,
 \qquad a_k(v)=\xi_k^2v.
\]
The drift part of the local divergence is
\[
 \frac{(A_k-B_kv)^2}{2\xi_k^2v}.
\]
Taking the stationary expectation and using \eqref{eq:CIR_stationary_moments} gives the first term of \eqref{eq:CIR_div}.  In one dimension,
\[
 \BW^2(\xi^2v,\xi_k^2v)=(\xi-\xi_k)^2v,
\]
whose stationary expectation gives the second term.  The rigorous passage from the localized uniformly elliptic objectives follows from \Cref{prop:CIR_localization}.
\end{proof}

\begin{corollary}[Different volatility of volatility and common mean reversion]\label{cor:CIR_xi}
If $\kappa_k=\kappa_0$ and $\theta_k=\theta_0$ for all $k$, and the candidate is restricted to the same $(\kappa_0,\theta_0)$, then
\begin{equation}\label{eq:xi_bary}
 \xi^*=\sum_{k=1}^K\pi_k\xi_k.
\end{equation}
The minimized divergence rate is
\begin{equation}\label{eq:xi_dispersion}
 \frac{\lambda\theta_0}{2}\sum_k\pi_k(\xi_k-\xi^*)^2.
\end{equation}
\end{corollary}

Embedding \eqref{eq:CIR_k} in the stochastic-volatility model of \cite{Heston1993},
\[
 \frac{\dd S_t}{S_t}=r\dd t+\sqrt{V_t}\dd W_t^S,
 \qquad \dd\langle W^S,B\rangle_t=\rho\dd t,
\]
gives the joint instantaneous covariance
\[
 v\begin{pmatrix}
 1&\rho\xi\\
 \rho\xi&\xi^2
 \end{pmatrix}.
\]
Experts may therefore disagree simultaneously about volatility of volatility and leverage.  At fixed $v$, the noncommuting matrix problem is rigorous on the bounded affine stochastic-volatility slice of \Cref{cor:Heston_slice}; no common-covariance or simultaneous-diagonalization assumption is required.

\begin{remark}[Scope of the square-root stochastic-volatility result]\label{rem:Heston_scope}
The dynamic candidate-first Hamilton--Jacobi--Bellman theory in \Cref{thm:HJB} assumes uniform ellipticity and therefore does not yet cover a full square-root stochastic-volatility barycenter up to the boundary $v=0$.  The results established here are the localized stationary square-root divergence, existence on compact parameter sets separated from the Feller boundary, and the fixed-$v$ noncommuting covariance selector.  A dynamic candidate-first barycenter for square-root experts would require a degenerate Hamilton--Jacobi--Bellman equation or controlled-martingale-problem extension with boundary analysis and is left open.
\end{remark}

\section{Conclusion}\label{sec:conclusion}

Aggregating diffusion models with different covariances requires a choice of approximation loss as well as a choice of KL orientation. We retain the covariance normalization of Gaussian drift information and measure shock disagreement by quadratic transport in a fixed state metric. The resulting criterion has an exact Euler-kernel construction, is equivariant under linear coordinate changes with the metric transformed accordingly, and permits a direct interpretation of the covariance penalty through a disagreement budget.

The two orientations yield distinct diffusion aggregates. Expert-first projection uses fixed expert state distributions, posterior-mean drift, and a covariance correction for drift dispersion. Candidate-first aggregation changes its own state distribution and is governed by a control problem. Its covariance selector has a matrix congruence representation with exact unrestricted threshold $D^2V+\lambda M\succ0$ in the original state coordinates. The scalar results show why averaging volatilities alone misses part of either problem: disagreement about drift can itself alter the selected volatility.

Three extensions remain. A general exact-transition expansion would identify the drift--covariance terms beyond the Ornstein--Uhlenbeck correction established here. Dynamic square-root models require a degenerate HJB analysis at the volatility boundary. For applications, covariance budgets and martingale restrictions can be imposed together; quantifying the resulting errors for path-dependent payoffs would clarify how to choose a budget for a particular pricing or control task.

\subsection*{Funding}
This work was partially supported by the Grant Agency of the Czech Republic under grant 24-11146S.

\subsection*{Conflicts of Interest}
The author declares no competing interests.

\section*{Declaration of generative AI and AI-assisted technologies in the manuscript preparation process}
During the preparation of this work, the author used ChatGPT and Claude for drafting, language editing, organization, and checks of the mathematical exposition. The author reviewed and edited the resulting material and takes full responsibility for the content of the publication.

\appendix

\section{The scalar martingale-constrained cubic}\label{app:cubic}

Differentiating \eqref{eq:martingale_constrained} with respect to $\sigma$ gives
\[
 \sum_k\pi_k\left\{
 \frac{\sigma(\sigma^2-\sigma_k^2)}{2\sigma_k^2}
 +\lambda(\sigma-\sigma_k)
 \right\}=0.
\]
Multiplication by two and collection of powers gives \eqref{eq:cubic_sigma}.  Put
\[
 a:=\sum_k\frac{\pi_k}{\sigma_k^2}>0,
 \qquad
 \bar\sigma:=\sum_k\pi_k\sigma_k>0,
 \qquad
 f_\lambda(\sigma):=a\sigma^3+(2\lambda-1)\sigma-2\lambda\bar\sigma.
\]
For every $\lambda>0$, $f_\lambda(0)<0$ and $f_\lambda(\sigma)\to+\infty$ as $\sigma\to\infty$.  If $\lambda\ge1/2$, then $f_\lambda$ is strictly increasing on $(0,\infty)$.  If $0<\lambda<1/2$, then $f_\lambda'$ is negative near zero and vanishes once, at
\[
 \sigma_{\min}=\sqrt{\frac{1-2\lambda}{3a}}.
\]
Hence $f_\lambda$ first decreases from an already negative value and then increases strictly to $+\infty$.  In either case it crosses zero exactly once on $(0,\infty)$.

Let $\sigma_\lambda$ denote that root.  For $\lambda>1/2$, the root equation and $a\sigma_\lambda^3\ge0$ imply
\[
 0<\sigma_\lambda\le \frac{2\lambda}{2\lambda-1}\,\bar\sigma,
\]
so $(\sigma_\lambda)$ is bounded for large $\lambda$.  Dividing \eqref{eq:cubic_sigma} by $2\lambda$ and passing to a convergent subsequence shows that every subsequential limit equals $\bar\sigma$; hence
$\sigma_\lambda\to\bar\sigma$.  As $\lambda\downarrow0$, every positive accumulation point solves
\[
 a\sigma^3-\sigma=0.
\]
The unique positive solution is $\sigma=a^{-1/2}$.  The roots cannot converge to zero because, for all sufficiently small $\lambda$, $f_\lambda$ remains negative on a fixed interval $(0,\varepsilon]$.  Therefore
\[
 \sigma_\lambda^2\longrightarrow a^{-1}
 =\left(\sum_k\frac{\pi_k}{\sigma_k^2}\right)^{-1},
\]
which proves \eqref{eq:cubic_small_lambda_limit}.

\end{document}